\documentclass[12pt, reqno]{amsart}
\usepackage[thm_section]{macros}
\usepackage{latexsym,amssymb}
\usepackage{graphics}
\usepackage{natbib}
\usepackage{bm}
\usepackage{bbm}
\usepackage{caption, subcaption}
\usepackage{epsfig}
\usepackage{color}
\usepackage{stackengine}
\usepackage{thm-restate}
\usepackage{dcolumn}

\newgeometry{margin=1.25in}

\newcommand\numberthis{\addtocounter{equation}{1}\tag{\theequation}}
\def\equiv{\mathrel{\ensurestackMath{\stackon[1.5pt]{=}{\scriptscriptstyle\Delta}}}}

\newcommand{\proj}{\operatorname{proj}}

\newtheorem{example}{Example}

\newcommand{\obs}{\mathrm{obs}}
\newcommand{\Tspace}{\bar{\mathcal T}}
\makeatletter
\patchcmd{\@settitle}{\uppercasenonmath\@title}{\large}{}{}
\patchcmd{\@setauthors}{\MakeUppercase}{\vspace{-1em}\normalsize}{}{}
\patchcmd{\section}{\scshape}{\bfseries}{}{}
\makeatother

\begin{document}

\pagestyle{plain} \bibliographystyle{ecca} \title{Empirical Bayes shrinkage (mostly) does
not
 correct \\ the measurement error in regression} \pagestyle{plain}
	
	\author{Jiafeng Chen, Stanford University \\ Jiaying Gu, University of
	Toronto \\ Soonwoo Kwon, Brown University 	\vspace{3em}}

      \thanks{Version:  \today. We thank Isaiah Andrews, Xiaohong Chen, John Friedman, Peter Hull,
       Patrick Kline, Chris Walters, and Merrill Warnick for helpful discussions. Kwon acknowledges support from the Annenberg Institute. Juan
       Yamin Silva provided excellent research assistance.}

	\begin{abstract}
		In the value-added literature, it is often claimed that regressing on empirical
		Bayes shrinkage estimates corrects for the measurement error problem in linear
		regression. We clarify the conditions needed; we argue that these conditions are
		stronger than the those needed for classical measurement error correction, which
		we advocate for instead. Moreover, we show that the classical estimator cannot be
		improved without stronger assumptions. We extend these results to regressions on
		nonlinear transformations of the latent attribute and find generically slow
		minimax estimation rates.
	\end{abstract}
	\maketitle

\newpage 
	\section{Introduction}

Heterogeneity of individuals is a vital element in many important areas of inquiry within
economics. Empirical Bayes methods \citep{robbins56} are applicable in many such settings
to denoise individual fixed effects from noisy data. These methods are increasingly
widely applied: Researchers use them to estimate individual effects of teachers (\citealp
{kane2008does, chetty2014measuringa, gilraine2020new}), mobility of
geographies (\citealp{chetty2018impacts}), value-added of hospitals (\citealp
{chandra2016health}), skill of patent examiners (\citealp{feng2020crafting}), quality of
managers (\citealp{fenizia2022managers}), and individual income dynamics (\citealp
{gu2017unobserved}), among others.

Researchers also often hope to quantify how unobserved individual attributes, like a
teacher value-added, predicts downstream economic outcomes, like long-term student
outcomes. In such settings, there is a common intuition in the empirical
literature that empirical Bayes shrinkage provides a correction for attenuation caused by
statistical noise in a linear regression estimator \citep
{jacob2005principle,kane2008estimating,chetty2014measuringa,angrist2023methods}. For
instance, \citet{angrist2023methods} write ``shrinkage corrects measurement error in
models that treat school value-added as a regressor,'' and this intuition appears
widespread. On the other hand, the classical literature in errors-in-variable regression
offers simple corrections for measurement error that apply in these settings
\citep{fuller2009measurement}, and it is unclear how including shrinkage estimates on the
right-hand side of a regression compares to the classical approach.%

This paper clarifies the conditions under which the regress-on-shrinkage estimator
corrects for measurement error.  Suppose the researcher would like to compute an
\emph{infeasible regression} of $Y_i$ on $\mu_i$, but only has access to noisy
measurements $ (X_i, \sigma_i)$, where $X_i \sim \Norm(\mu_i,
\sigma_i^2).$ Commonly, researchers fit the regression of $Y_i$ on $\hat\mu_i(X_i,
\sigma_i)$, where $\hat\mu_i(X_i,
\sigma_i)$ is a linear shrinkage estimate. Importantly, an often-ignored  condition for
 the consistency of this \emph{regress-on-shrinkage estimator} is that the degree of
 shrinkage---implicitly a function
 of the noise level $\sigma_i$---does not correlate with the outcome $Y_i$ once $\mu_i$ is
 controlled. Under such a condition, adding flexible functions in $\sigma_i$ to the
 infeasible regression would not change the coefficient. This is a condition much like
 \emph{precision independence} \citep{walters2024empirical,chen2023empirical} and can be
 economically unrealistic. On the other hand, in this setting, because the $\sigma_i^2$
 are observed, the attentuation bias in the regression of $Y_i$ on $X_i$ is
 estimable---and thus can be directly corrected. Such classical measurement error
 corrections does not impose any precision independence assumption. This makes the
 classical corrections more robust and---in our view---preferable.

Second, we show that without stronger assumptions, this classical estimator is effectively
the \emph{only} reasonable estimator for the regression coefficient in the infeasible
regression, up to asymptotic equivalence; it is thus (vacuously) semiparametrically
efficient. This result shows that the deficiencies of the regress-on-shrinkage
estimator are not due to insufficiently flexible empirical Bayes modeling. No amount of
flexible modeling---despite tools proposed by, e.g., \citet
{chen2023empirical,gilraine2020new,kwon2023optimalshrinkageestimationfixed}---can improve
on the classical measurement error correction.

Third, we extend our results to more complex settings, where the infeasible regression
involves known nonlinear transforms $f(\mu_i)$ (e.g., indicators like $\one(\mu_i >
\mu_0)$ may represent ``high value-added'' teachers). Under strong assumptions---that both
 the infeasible regression and parametric empirical Bayes models are correctly
 specified--–a regression of $Y_i$ on $\E[f(\mu_i) \mid X_i, \sigma_i]$ recovers the
 infeasible regression coefficients. However, we caution that these strong assumptions are
 difficult to relax without relinquishing appealing features of the resulting estimator,
 due to the fundamental statistical difficulty of deconvolution \citep{cai2011testing}. To
 that end, we show that if the infeasible regression coefficient is estimable at
 polynomial $(n^{-\alpha})$ rates of convergence uniformly across data-generating
 processes that impose few assumptions, then $f (\cdot)$ must necessarily be an analytic
 function---an extreme smoothness requirement. Put differently, if $f(\cdot)$ is not
 smooth, then there is some data-generating process under which one would need
 exponentially-in-$k$ larger sample sizes to reduce uncertainty by a factor of $k$.

  We compare these methods in simulation and across two empirical applications. The
  simulation results confirm that the classical estimator performs well, and substantially
  better than the regress-on-shrinkage estimator, across a wide range of data generating
  processes. Regress-on-shrinkage estimates can be severely biased even under settings
  that are calibrated to real data. Moreover, we confirm that regressions involving
  nonlinear transformations $f(\mu_{i})$ are indeed difficult, with widely dispersed
  estimates even with a reasonably large sample size. 

  We then compare the methods on two empirical applications. The first one revisits
  \citet{feng2020crafting} which concerns about patent examination practise on patent
  outcomes. The second application revisits \citet{bau2020teacher}, who analyze teacher
  value added in a low income country. We find potentially substantive economic
  differences when using the classical estimator rather than the regress-on-shrinkage
  estimator in one of the applications. We also find supportive evidence that the
  conditions needed for the regress-on-shrinkage estimator to provide reliable estimates
  are violated.

This paper is related to the classical errors-in-variable regression literature
\citep{fuller2009measurement,bickel1987efficient} as well as a recent literature on
generated regressors. We highlight and compare a few. \citet{rose2022effects} study
multidimensional teacher value-added (e.g., value-added on criminal justice event, on
math performance, etc.); they advocate for estimating the variance-covariance matrix of
teacher value-added by correcting for measurement error, rather than by taking the
variance-covariance matrix of empirical Bayes posterior means. We show that the same
extends to regressions of downstream variables on value-added. \citet{deeb2021framework}
studies the regress-on-shrinkage estimator and proposes corrections to its standard
error that accounts for the uncertainty in estimating empirical Bayes hyperparameters.
In a similar setting, \cite{xie2025automatic} establishes conditions under which
regression-on-shrinkage estimators automatically yield valid inference without requiring
additional adjustments.  \citet{battaglia2024inference} study the generated regressor
problem with machine learning predictions for $\mu_i$. The setting is related but
distinct---they do not impose the Gaussian structure commonly imposed in the empirical
Bayes literature. To our knowledge, our efficiency and minimax rate results have not
appeared in the literature.

This paper proceeds as follows. \Cref{sec:ols} studies regression on latent $\mu_i$. 
\Cref{sec:nonlinear} presents our results for regression on $f(\mu_i)$ with nonlinear $f
(\cdot)$. \Cref{sec:simulation,sec:empirical} illustrate our results using a simulation
and two empirical applications.

\section{Linear regression coefficients}
\label{sec:ols}

Suppose we observe $(Y_i, X_i, \sigma_i)_{i=1}^n$ for a sample of individuals. Throughout,
we will refer to these individuals as teachers, but the applications extend beyond
teacher value-added. Here, $Y_i$ is some outcome variable, $X_i$ is a noisy measure of an
unobserved teacher attribute $\mu_i$, and $\sigma_i$ is the observed standard error for
$X_i$. For instance, $Y_i$ would be some attribute of a teacher $i$ (e.g. teacher-level
mean of student outcomes), $X_i$ would be the estimated teacher value-added of $i$,
$\mu_i$ is the true teacher value-added for teacher $i$, and $\sigma_i$ is the estimated
standard error for $X_i$. For expositional simplicity, our main results shall restrict to
setups of the data that aggregate to the teacher level. \Cref{sub:implement} discusses
implementations of analogous approaches with disaggregated data.

Following the empirical Bayes literature, we assume $X_i \mid
Y_i, \mu_i, \sigma_i \sim \Norm(\mu_i, \sigma_i^2)$ is Normally distributed
and unrelated to $Y_i$, collected in the following assumption:
\begin{as} \label{basicassp}
We assume $(Y_i, X_i, \mu_i, \sigma_i) \iid P_0$. We impose the following assumptions on $P_0$: 
	\begin{enumerate}
		\item $X_i \mid Y_i, \mu_i, \sigma_i \sim \Norm(\mu_i, \sigma_i^2)$
		\item $\sigma_\mu^2 \equiv \var(\mu_i) > 0$ and $\E[\sigma_i^2] > 0$.
		\item $P_0(\sigma_i > 0) = 1$. 
	\end{enumerate}
	\end{as}

\Cref{basicassp}(2) and (3) are standard regularity assumptions. \Cref{basicassp}(1) is
common in the empirical Bayes literature. The Normality of $X_i$ is motivated by the fact
that $X_i$ is typically an estimate of $\mu_i$ with micro-data within a teacher: For
instance, $X_i$ may be a teacher-level mean of student scores, and the central limit
theorem provides a Gaussian approximation, where $\sigma_i^2$ is the estimated standard
error \citep[see][]{walters2024empirical}.

\Cref{basicassp}(1) also implies that the outcome $Y_i$ does not predict the noise
 component of $X_i$, i.e., $X_i \indep Y_i \mid \mu_i, \sigma_i^2$. This assumption may be
 violated if, for instance, $X_i$ is the sample mean of student test scores at the teacher
 level, but $Y_i$ is the average downstream outcome of that same set of students.
 Meanwhile, if $X_i$ and $Y_i$ come from different sets of students (and we assume student
 outcomes are independent) or if $Y_i$ is an underlying characteristic of teacher
 $i$,\footnote{For instance, \citet{chandra2016health} consider a regression of hospital
 $i$ market size ($Y_i$) on hospital quality ($\mu_i$), which are then estimated from
 clinical outcome of patients ($X_i$). Here, it is reasonable to assume that hospital size
 is independent from the noise component of patient outcomes $X_i - \mu_i$. } then this
 assumption is reasonable. \Cref{sub:implement} considers general cases where $Y_i$ needs
 not be independent of the noise in $X_i$.

Empirical researchers are often interested in the population \emph{infeasible} regression
of $Y_i$ on $\mu_i$: \[
Y_i =
\alpha_0 + \beta_0 \mu_i + \eta_i, 
\numberthis
\label{model}
\] 
where the population regression coefficient $\beta_0 = \frac{\cov(Y_i, \mu_i)}{\var
(\mu_i)}$ by
definition. Here, we do not treat \cref{model} as a linear restriction on $\E[Y_i \mid
\mu_i]$, but simply as a definition for $\beta_0$. That is, $\beta_0$ is the OLS
 regression coefficient one would have obtained had one access to the true unobserved
 attribute $\mu_i$ and infinitely many observations. Regressions of this form appear in,
 among others, \citet
 {chandra2016health,
 chetty2014measuringa,jacob2005parents,jackson2018test,warnick2024instructor,mulhern2023beyond}.\footnote{We take as given that the infeasible regression coefficient \eqref{model} is the target
parameter. It is possible, however, that empirical researchers may prefer other
estimands. For instance, since decisions are only functions of $(X_i,
\sigma_i)$, we might be interested instead in the statistical relationship between $Y_i$
 and functions of $(X_i, \sigma_i)$. For instance, we might form a prediction $\hat\mu_i
 (X_i, \sigma_i)$ and we might want to assess how predictive
 \emph{the predicted teacher value-added} is for $Y_i$, instead of how predictive
 \emph{the true} $\mu_i$ is for $Y_i$. The former is more relevant, say, if we are more
  interested in assessing the quality of feasible predictions for teacher effects, rather
  than the inherent relationship between true teacher effects and outcomes.}

For instance, when $Y_i$ is the mean student  long-term outcome for those taught by
teacher $i$ and $\mu_i$ is a teacher value-added for a teacher experienced by student
$i$, then $\beta_0$ is the coefficient of the best linear prediction the long-term
outcome from true teacher value-added. Under appropriate identifying assumptions such
that $Y_i$ is unbiased for the mean potential outcome of students assigned to teacher
$i$, $\beta_0$ admits a causal interpretation as the best linear approximation to the
conditional mean of teacher causal effects on true teacher value added. %

Immediately, \eqref{model} implies that $\beta_0$ is identified
by the following formula. 
\begin{prop}
Under \cref{basicassp}, $\beta_0$ is equal to the following function of the joint
distribution of the observed data $(Y_i, X_i, \sigma_i)$:
\[
\beta_0 = \frac{\cov(Y_i, X_i)}{\var(X_i) - \E[\sigma_i^2]} = \underbrace{\frac{\cov(Y_i,
X_i)}
{\var
(X_i)}}_{\text{Regression coefficient of $Y_i$ on $X_i$}} \underbrace{\frac{\var(X_i)}
{\var (X_i) - \E
[\sigma_i^2]}}_{\text{Inflation factor}}.
\numberthis 
\label{eq:inflation}
\]
\end{prop}
\begin{proof}
The result follows immediately from the observations that (a) $\cov(Y_i, X_i) =
\cov
(Y_i, \mu_i)$ and (b) $\var(\mu_i) = \var(X_i) - \E[\sigma_i^2]$.
\end{proof}

\Cref{eq:inflation} suggests a simple analogue estimator for $\beta_0$ that replaces
population covariances, variances, and expectations with their sample
counterparts:\footnote{We shorthand $\E_n[W_i] = \frac{1}{n} \sum_{i=1}^n W_i$, $\var_n
(W_i) = \E_n[W_i^2] - (\E_n[W_i])^2$, and $\cov_n(W_i, Z_i) = \E_n[W_i Z_i] - \E_n[W_i]
\E_n[Z_i]$.} \[
\hat\beta = \frac{{\cov}_n(Y_i, X_i)}{{\var}_n(X_i) - {\E}_n
[\sigma_i^2]}. \numberthis \label{eq:inf_est}
\] 
Under standard conditions, $\hat\beta$ is consistent and asymptotically Normal, whose
 asymptotic distribution can be consistently estimated by a nonparametric bootstrap given
 $(Y_i, X_i, \sigma_i)$. This estimator is classical from the errors-in-variable
 regression literature (e.g., Section 3.1 of \citealp{fuller2009measurement}, (23) in 
 \citealp{deaton1985panel}) and more recently advocated by \citet
  {de2024estimating}. Moreover, the consistency and Normality of $\hat\beta$ does not
  require $X_i$ to be Normally distributed. 

Our first main point is to advocate for this classical estimator $\hat\beta$ as opposed to
common estimation approaches for \eqref{model}. A popular---and standard---alternative
estimator of $\beta_0$ is the regression coefficient of $Y_i$ on estimated empirical Bayes
posterior means $\hat\mu_i(X_i,\sigma_i)$, following a parametric empirical Bayes
procedure:
 \begin{align}
\tilde \beta &\equiv \frac{\cov_n(Y_i, \hat\mu_i)}{\var_n(\hat\mu_i)}
 \label{eq:reg_shr_est}\\
 \hat \mu_i (X_i, \sigma_i) &\equiv \frac{\sigma_i^2}{\sigma_i^2 + \hat\sigma_\mu^2} \hat\mu
+
\frac{\hat\sigma_\mu^2}{\sigma_i^2 + \hat\sigma_\mu^2} X_i
\label{eq:par_emp_bayes} \\
\text{ where }
\hat\mu &\equiv \E_n[Y_i] \text{ and } \hat \sigma_\mu^2 \equiv \var_n(Y_i) - \E_n
[\sigma_i^2]. \nonumber
 \end{align} This approach is followed in, e.g., \citet
  {jacob2005parents, kane2008estimating, warnick2024instructor, jackson2018test,
  bau2020teacher,angelova2023algorithmic}. The estimator $\tilde \beta$ is widely thought to be consistent for
  $\beta_0$, which we argue rests strong assumptions, often implicitly imposed.\footnote{%
\citet{jacob2005parents} (Appendix C) states that ``one can easily show that using the
EB estimates as an explanatory variable in a regression context will yield point estimates
that are unaffected by the attenuation bias that would exist if one used simple OLS
estimates.''

\citet{angrist2023methods} write that ``shrinkage corrects measurement error in models
 that treat school value-added as a regressor. Putting the unbiased but noisy estimate
 [in our notation, $X_i$] on the right-hand side of a regression results in attenuation
 bias toward zero due to classical measurement error; the posterior mean introduces
 non-classical measurement error that corrects this so that a regression with [$\mu_i
 (X_i, \sigma_i)$] on the right yields the same coefficient as using the true
 [$\mu_i$].'' } We also note that when $\sigma_i^2 = \sigma^2$ are constant for all $i$,
 then $\tilde \beta = \hat\beta$, but they are no longer equal in heteroskedastic
 settings.

  To this end, we show that, first, the estimator $\tilde \beta$ is consistent for
  $\beta_0$ only under much stronger conditions than \cref{basicassp}, rendering it less
  robust and less preferable to $\hat\beta$. Second, we show that it is impossible to
  improve upon $\hat\beta$ with any other procedure, including more flexible empirical
  Bayes procedures \citep
  {kwon2023optimalshrinkageestimationfixed,gilraine2020new,chen2023empirical}, at least
  without imposing stronger assumptions. Under \cref{basicassp}, any consistent and
  asymptotically Normal estimator of $\beta_0$ is in fact asymptotically equivalent to
  $\hat\beta$. Thus, any empirical Bayes procedure can at most match the performance of
  $\hat\beta$.

\subsection{When is $\tilde \beta$ consistent?}

Since $\mu \equiv \E[\mu], \sigma_\mu^2\equiv \var(\mu)$ are consistently estimable, let
us restrict attention to an
oracle counterpart of $\tilde \beta$ where $\mu,\sigma_\mu^2$ are known, namely \[
\tilde\beta^* \equiv \frac{\cov_n(Y_i, \mu_i^*)}{\var_n(\mu_i^*)} \text{, with population
limit } \tilde\beta_0\equiv \frac{\cov_{P_0}(Y_i, \mu_i^*)}{\var_{P_0}(\mu_i^*)}.
\]

We first formalize the folklore in the empirical literature on the consistency of $\tilde
\beta$.
\begin{as}
\label{as:strong_high_level}
For $\mu_i^* = \mu_i^*(X_i, \sigma_i)$ the oracle counterpart to 
\eqref{eq:par_emp_bayes}, the distribution
$P_0$
satisfies
\begin{enumerate}
	\item (Forecast unbiasedness) $\cov(\mu_i^*, \mu_i) = \var(\mu_i^*)$ 
	\item (Exogeneity) $\E[\eta_i \mu_i^*] = 0$.
\end{enumerate}
\end{as}
\begin{prop}
\label{prop:consistency_eb}
Under \cref{as:strong_high_level,basicassp}, $\tilde \beta_0 = \beta_0$.
\end{prop}

\Cref{as:strong_high_level}(1) is often referred to as forecast unbiasedness
\citep{chetty2014measuringa,chetty2014measuringb,stigler19901988}. It assumes that the
empirical Bayes posterior means are reasonable predictors of the true unobserved
attribute, in the sense that a hypothetical regression of the unobserved attribute $\mu_i$
on its predicted value $\mu_i^*$ returns a regression coefficient of $1$.\footnote{When
the empirical Bayes prior is well-specified, i.e. $\mu_i^* = \E[\mu_i \mid X_i,
\sigma_i]$, we have that $\E[\mu_i
\mid
\mu_i^*] = \mu_i^*$ by the law of iterated expectations, and 
\cref{as:strong_high_level}(1) is satisfied. For Gaussian-prior empirical Bayes, even if
 the empirical Bayes model is misspecified, \cref{as:strong_high_level}(1) would be true
 if we only assume
\eqref{eq:moment_indep}, meaning that $\sigma_i$ does not predict the first two moments of
$\mu_i$.} 

\Cref{as:strong_high_level}(2) is a key assumption that seems to be implicitly taken for
granted in the literature. Importantly, the condition $\E[\eta_i X_i] = 0$---which we
impose by definition of $\beta_0$ as the population projection coefficient---does not on
its own justify \cref{as:strong_high_level}(2).\footnote{Under homoskedasticity, i.e.
$\sigma_i = \sigma$ for all $i$, $\mu_i^*$ is simply a linear function of $X_i$, and
$\E[\mu_i^* \eta_i] = 0$ holds by construction; however, the same is not true when
$\sigma_i$ are heterogeneous and can be correlated with $Y_i$.} The reason is that
$\mu_i^*$ is a function of both $X_i$ and $\sigma_i$, and we have made no assumptions on
how $\sigma_i$ interacts with $\eta_i$.

Under \cref{basicassp,as:strong_high_level}, we review and formalize two arguments in the
literature for \cref{prop:consistency_eb}, and point out where each of the assumptions in
\cref{as:strong_high_level} is used. In particular, we have not found any work in the
literature that explicitly mentions \cref{as:strong_high_level}(2) or sufficient
conditions for it, even though---as we shall see---\cref{as:strong_high_level}(2) is
crucial for the consistency of $\tilde\beta$.

We formalize the two main arguments used to justify the consistency of $\tilde{\beta}$.  First, many papers justify $\tilde \beta_0 = \beta_0$
using an argument from
\citet{jacob2005parents}.

\begin{proof}(\citet{jacob2005parents}'s argument for \cref{prop:consistency_eb})
	Observe that for $v_i =
\mu_i - \mu_i^*$, under \cref{as:strong_high_level}(1), \[
\mu_i = \mu_i^* + v_i \quad \E_{P_0}[v_i \mu_i^* ] = 0.
\]
As a result, we can rewrite \eqref{model} as $Y_i = \alpha_0 + \beta_0 \mu_i^* + \beta_0
v_i + \eta_i$, $\E_{P_0}[\eta_i X_i] = 0$. 
By \cref{as:strong_high_level}(2), \[
\E[(\beta_0 v_i + \eta_i) \mu_i^*] = \beta_0 \E[v_i \mu_i^*] + \E[\eta_i \mu_i^*] = 
\colorbox{Red!5}{$\E
[\eta_i\mu_i^*] = 0$}.
\]
As a result, regressing $Y_i$ on $\mu_i^*$ recovers the coefficient $\beta_0$. 
\end{proof}

A second popular intuition justifies $\tilde\beta_0 = \beta_0$ by appealing to
instrumental variables \citep[][pp.~2639]{chetty2014measuringb}.

\begin{proof}(IV argument for \cref{prop:consistency_eb}) We can write $
 (Y_i, \mu_i, \mu_i^*)$ as a two-stage least-squares specification, where $\mu_i$ is an
 ``endogenous treatment'' and $\mu_i^*$ is an ``exogenous instrument'':
\begin{align*}
Y_i &= \alpha_0 + \beta_0 \mu_i + \eta_i \\
\mu_i &= \gamma_0 + \pi_0 \mu_i^* + v_i.
\end{align*}
\Cref{as:strong_high_level}(1) implies that the population first-stage coefficient is one:
 $\pi_0 = 1$. Crucially, \cref{as:strong_high_level}(2) is exactly the exogeneity and
 exclusion assumption, implying that $\mu_i^*$ is a valid instrument. Therefore, $\beta_0$
 is equal to the population two-stage least-squares coefficient, which is further equal to
 the reduced-form coefficient $\tilde \beta_0$ since the first-stage coefficient $\pi_0 =
 1$:
 \[\beta_0 =
\frac{\cov (Y_i, \mu_i^*)} {\cov
(\mu_i^*,
\mu_i)} = \frac{\overbrace{\cov(Y_i, \mu_i^*)/\var(\mu_i^*)}^{\text{Reduced-form coef.
of $Y_i$ on $\mu_i^*$}}} {\underbrace{\cov (\mu_i^*, \mu_i) / \var (\mu_i^*)}_{
\text{First-stage coef. of $\mu_i$ on $\mu_i^*$, $\pi_0 = 1$}}} = \frac{\cov (\mu_i^*,
\mu_i)}{
\var
(\mu_i^*)} = \tilde\beta_0.\]
\end{proof}

\Cref{as:strong_high_level} is a set of high-level assumptions imposed to justify $\tilde
\beta_0
= \beta_0$. The following assumptions are natural sufficient conditions for
\cref{as:strong_high_level}. Though they are stronger than \cref{as:strong_high_level},
scenarios that satisfy \cref{as:strong_high_level} but violate \cref{as:strong_low_level}
are economically knife-edge. Moreover, all assumptions in \cref{as:strong_low_level} are
testable.
\begin{as}[Sufficient conditions for \cref{as:strong_high_level}]
\label{as:strong_low_level}
\leavevmode\vspace{0.1\baselineskip}
\begin{enumerate}
\item (Precision independence) $\sigma_i \indep (\mu_i, Y_i)$ under $P_0$

	\item (Precision independence in first two moments) The distribution $P_0$ satisfies
	\[\E_{P_0}[\mu_i \mid \sigma_i] = \mu \text{ and } \var_
{P_0}
[\mu_i \mid \sigma_i] = \sigma_\mu^2, \numberthis \label{eq:moment_indep}\]
	\item (Linearity of the conditional expectation function and exogeneity of $\sigma$)
	$\E_ {P_0} [Y_i \mid \mu_i,
	\sigma_i] = \alpha_0 + \beta_0 \mu_i$. 
	
\end{enumerate}
\end{as}

\begin{restatable}{lemma}{lemmaimplication}
\label{lemma:implication}
Under \cref{basicassp}, \cref{as:strong_low_level}(1) implies \cref{as:strong_high_level};
\Cref{as:strong_low_level}(2) implies \cref{as:strong_high_level}(1), and
\cref{as:strong_low_level}(3) implies \cref{as:strong_low_level}(2).
\end{restatable}

A strong sufficient condition for \cref{as:strong_high_level} is directly that $\sigma_i$
is completely independent of the joint distribution of $(Y_i, \mu_i)$. 

 \Cref{as:strong_low_level}(2)--(3) are in some respects weaker.
\Cref{as:strong_low_level}(2) imposes that $\sigma_i$ does not predict $\mu_i$ at least in
its first two conditional moments. This is a moment version of the precision independence
condition discussed in \citet{chen2023empirical}, who shows examples in which
\cref{as:strong_low_level}(2) fails to hold. This precision independence assumption alone
does not involve $Y_i$, and is thus insufficient to justify $\tilde\beta_0 = \beta_0$.

\Cref{as:strong_low_level}(3) imposes two strong assumptions on $P_0$. The first is that
the conditional expectation function of $Y_i$ given $\mu_i$ is in fact linear in $\mu_i$,
viewing $\beta_0$ instead as the slope of that linear relationship rather than the
population best linear approximation of $\E[Y_i \mid \mu_i]$. Second,
\cref{as:strong_low_level} imposes that $\sigma_i$ does not have additional predictive
power over $Y_i$ given $\mu_i$. This is the analogue of a precision independence
assumption for the outcome variable $Y_i$, and can fail due to similar reasons outlined in
\citet{chen2023empirical}. This assumption is rejected in our empirical application.

Since $\hat\beta$ is consistent for $\beta_0$ without imposing \cref{as:strong_high_level}
or \cref{as:strong_low_level} at all, it should be preferred over the estimator
$\tilde\beta$ on robustness grounds. Nevertheless, motivated by
\cref{as:strong_low_level}(2), we might wonder whether the defects of the
regress-on-shrinkage estimator $\tilde \beta$ is due to insufficiently flexible empirical
Bayes modeling. Perhaps the Normality assumption in \cref{basicassp} can be exploited to
yield efficiency benefits. The next subsection answers this question in the negative.
Specifically, we show that under \cref{basicassp}, not only is $\hat\beta$ a
semiparametrically efficient estimator for $\beta_0$, but all well-behaved consistent
estimators for $\beta_0$ must be asymptotically equivalent to $\hat\beta$. This provides
strong justification for $\hat\beta$ as the \emph{only} estimator for $\beta_0$
asymptotically.

\subsection{Efficiency and uniqueness of $\hat\beta$}

Let $\mathcal P_0$ be the set of distributions $P_0$ that satisfy
\cref{basicassp} as well as some technical conditions, stated in \cref{asec:eff}. Let
$\mathcal P$ be the set of distributions on the observed data $(Y_i, X_i, \sigma_i)$ that
is induced by members of $\mathcal P_0$. 

Following standard semiparametric theory
\citep{van2000asymptotic,tsiatis2006semiparametric}, we restrict our attention to
\emph{regular and asymptotically linear} (RAL) estimators, defined formally in
\cref{asec:eff}. RAL estimators are asymptotically Normal along all local perturbations of
$P_0$ within $\mathcal P_0$. Indeed, most $\sqrt{n}$-consistent and asymptotically Normal
estimators are RAL, and \emph{semiparametrically efficient} estimators are exactly the
optimal estimators \emph{among the class} of RAL estimators.

We also recall Definition 2.2 in \citet{chen2018overidentification} on local just
identification, again formally stated in \cref{asec:eff}. Local just identification is the
semiparametric analogue to just identification in parametric GMM models. In a just
identified GMM model, there are no additional moment conditions to exploit, and all GMM
weightings yield the same estimator. Exactly analogously, in just identified
semiparametric models, there are no (local) testable restrictions to exploit. Just like
all weighted GMM estimators are exactly the same under just identification, under local
just identification, all RAL estimators are asymptotically equivalent. They have the same
asymptotic distribution, and they are all (vacuously) semiparametrically efficient.

The next theorem shows that many members in $\mathcal P$ are locally just identified by
$\mathcal P$. The results of \citet{chen2018overidentification} then imply that all RAL
estimators for the regression coefficient $\beta_0$ are asymptotically equivalent to each
other and to $\hat\beta$, which we also verify in the following theorem. This also implies
that $\hat\beta$ is semiparametrically efficient, though in a vacuous sense.

\begin{restatable}{theorem}{thmeff}
\label{thm:eff}
Let $P \in \mathcal P$ and let $P_0 \in \mathcal P_0$ be the corresponding distribution
over the complete data, defined formally in \cref{asec:eff}. %
Then $P$ is
locally just
identified by $\mathcal P$. As a result,
\begin{enumerate}
	\item Any RAL estimator $\check\beta$ for
\[\beta_0 \equiv \beta_0(P) \equiv \frac{\cov_P(X,Y)}{\var_P(X) - \E_P[\sigma^2]}\] is
asymptotically
equivalent to the analogue estimator $\hat\beta$: $
\sqrt{n} (\check\beta - \hat\beta) = o_P(1).
$
\item The semiparametric efficiency bound for $\beta_0(P)$ is equal to the asymptotic
variance of $\hat\beta$ at $P$. 
\end{enumerate}

\end{restatable}

\Cref{thm:eff}, an application of Example 25.35 in \citet{van2000asymptotic}, rules out
efficiency gains from alternative estimators under the minimal assumptions
\cref{basicassp}. As a result, $\hat\beta$ is strongly justified for estimating $\beta$,
since it is in fact the unique estimator for $\beta$ in an asymptotic sense. The
implications of \cref{thm:eff} are not limited to the regression coefficient $\beta_0$;
indeed, any regular parameter of $P$ admits unique RAL estimators in the sense of
\cref{thm:eff}(1).

Imposing stronger assumptions than \cref{basicassp} amounts to shrinking the model
$\mathcal P$. Doing so would generally result in testable overidentification restrictions
and weakly decrease the efficiency bound. Under \cref{as:strong_low_level}(1) or similar
assumptions, for instance, the estimator $\tilde \beta$ is indeed more efficient than
$\hat\beta$ \citep{sullivan2001note}, though yet more efficient estimators exist
\citep{bickel1987efficient}. As discussed, \cref{as:strong_low_level} is much stronger
than \cref{basicassp}.

One assumption that is arguably reasonable is the conditional Normality structure of
$Y_i$, given that $i$ is a teacher and sometimes $Y_i$ is a sample mean of student
outcomes at the teacher level. In particular, we might consider imposing $Y_i \mid
\theta_i, X_i, \mu_i, \sigma_i, \nu_i \sim \Norm(\theta_i, \nu_i^2)$ for some unknown
$\theta_i$ and known $\nu_i^2$, and assume that $(Y_i, \theta_i, \nu_i, X_i, \mu_i,
\sigma_i) \iid P_0^*$ under the additional conditional Normality assumption. Simple
modification to the proof of \cref{thm:eff} shows that such a restriction does not
alter the conclusion of \cref{thm:eff}.

\subsection{Implementation} \label{sub:implement}We close this section with a discussion
of implementing
$\hat\beta$ in richer environments. First, given teacher-level data $(Y_i, X_i, Z_i,
\sigma_i)$, we may want to include covariates $Z_i$ in the infeasible regression: \[ Y_i =
\alpha_0 + \beta_0 \mu_i + \gamma_0' Z_i + \epsilon_i.
\]
Suppose $Z_i$ does not predict the noise component of $X_i$, then by
Frisch--Waugh--Lovell, \[
	Y_i = \beta_0 (\mu_i - \proj(\mu_i \mid Z_i)) + \epsilon_i
\]
and $X_i - \proj(X_i \mid Z_i) \sim \Norm(\mu_i - \proj(\mu_i \mid Z_i), \sigma_i^2)$,
where $\proj(\cdot \mid Z_i)$ is the population linear projection of a random variable
onto $Z_i$. 
Thus, our analysis above applies to the partialled out teacher effects $\mu_i- \proj(\mu_i \mid Z_i)$ and its
measurement $X_i - \proj(X_i \mid Z_i)$. 

More generally, it is common practice to consider a student-level regression. For $j$ a
student associated with teacher $i$, we consider the infeasible regression \[ Y_{ij} =
\alpha_0 +
\beta_0 \mu_j + \gamma_0'Z_{ij} + \epsilon_{ij}
\numberthis
\label{eq:infeasible_disagg}
\]
where $X_{ij}$ is unbiased for $\mu_j$.\footnote{One plausible and precise  sampling
process is as
follows. Suppose \[
	(N_i, \mu_i, (Y_{ij}, Z_{ij}, X_{ij})_{j=1}^{N_i}) \iid P_0
\]
and $P_0$ is such that $\E_{P_0}[X_{ij} \mid \mu_i, N_i] = \mu_i$. We also assume that
conditional on $\mu_i, N_i$, the student-level variables $(Y_{ij}, Z_{ij}, X_{ij})_{j=1}^
{N_i}$ are uncorrelated across $j$. 

To connect the aggregate and disaggregated setups, suppose $Z_{ij} = Z_i$ is not a
function of the student. Then the population OLS regression  \eqref{eq:infeasible_disagg}
is
equivalent to the teacher-level weighted least squares regression \[
	Y_i = \alpha_0 + \beta_0 \mu_j + \gamma_0' Z_i + \epsilon_{ij}
\]
where $Y_i = \frac{1}{N_i} \sum_{j} Y_{ij}$ and observations are weighted by $N_i$. 
} Here, we do not assume $X_ {ij}, Y_ {ij}, Z_ {ij}$
are uncorrelated within a teacher $i$, and so $Y,Z$ may predict the noise component of
$X$. This infeasible regression coefficient $\theta_0 = (\alpha_0, \beta_0, \gamma_0')'$
can be written as \[
	\theta_0 = \E_{P_0}\bk{
		\sum_{j=1}^{N_i} W_{ij} W_{ij}'
	}^{-1} \E_{P_0}\bk{\sum_{j=1}^{N_i} W_{ij} Y_{ij}}.
\]
The quantities $\E_{P_0}\bk{
		\sum_{j=1}^{N_i} W_{ij} W_{ij}'
	}$, $\E_{P_0}\bk{\sum_{j=1}^{N_i} W_{ij} Y_{ij}}$ involves infeasible terms such as
	$\E \sum_{j} \mu_i^2$, $\E \sum_j \mu_i Z_{ij}$, $\E\sum_j \mu_iY_{ij}$. Fortunately,
	they can be similarly estimated from $X_{ij}$: For instance, for $ X_i = \frac{1}
	{N_i} \sum_{j=1}^{N_i} X_{ij}$, \[
		\E\bk{\sum_{j=1}^{N_i} \mu_i Z_{ij}} = \E\bk{
			\sum_{j=1}^{N_i}  X_i Z_{ij}} - \E\bk{\underbrace{\frac{1}{N_{i}-1} \sum_{j=1}^{N_i}
			Z_ {ij} (X_
						{ij}
						- 
						X_i)}_{\text{Empirical covariance between $X_{ij}$ and $Z_{ij}$}}
		}.
	\]
	Thus, we may use $\sum_{j=1}^{N_i}  X_i Z_{ij} - \frac{1}{N_{i}-1} \sum_{j=1}^{N_i}
			Z_ {ij} (X_
						{ij}
						- 
						X_i)$ to substitute for $\sum_j \mu_i Z_{ij}$, and similarly for other terms
						in $\theta_0$. This construction is very similar to $\hat\beta$, since it
						essentially debiases an empirical moment with $X_i$ by adjusting for the
						impact of second moments. In a slightly different context, recent work by
						\citet{de2024estimating} proposes this estimator as well.

	Even more
	conveniently,
	it turns out that such an analogue can be simply implemented by instrumenting for $X_
	{ij}$ in following regression with a \emph{leave-one-out instrument} $X_{i,-j} =
	\frac{1}
	{N_i-1}
	\sum_{k\neq j} X_{ik}$
	\[
		Y_{ij} = \alpha_0 + \delta_0 X_{ij} + \gamma_0'Z_{ij} + \epsilon_{ij}.
	\]
\citet{devereux2007improved} shows\footnote{We thank Patrick Kline for this reference.}
that these two estimates are in fact numerically
	equivalent. This equivalence means that researchers can conveniently implement the
	measurement error correction by constructing the leave-one-out instrument and use
	off-the-shelf routines, without other
	multistep procedures.

\section{Regression coefficients involving nonlinear functions of $\mu_i$}
\label{sec:nonlinear}

Sometimes empirical researchers are interested in the projection coefficient on some known
nonlinear function $f(\cdot)$ of the latent quantities $\mu_i$:
\[
Y_i = \rho_0 + \tau_0 f(\mu_i) + \eta_i.
\]
For instance, we might be interested in how being a ``good teacher'' predicts outcomes,
and being a good teacher is defined as $\one(\mu_i > \mu_0)$ for some threshold $\mu_0$.
For an example of setting where $\beta_{0}$ is of interest under such specification, see
Section 3.2.2 of \cite{bruhn2022regulatory}.

Unfortunately, estimating $\tau_0$ involves some unpleasanat tradeoffs for the analyst, as
we no longer have access to a simple estimator like $\hat\beta$. On the one hand,
imposing a strong assumption does allow us to recover $\tau_0$ by regressing $Y_i$ on
the \emph{correctly specified} parametric empirical Bayes posterior means for $f
(\mu_i)$. The downside of this approach is that the assumptions involved can be quite
strong, and are unlikely to be justified. On the other hand, without these assumptions,
estimating $\tau_0$ is fundamentally difficult. This difficulty is in the sense that the
the error of the best possible estimator contracts at subpolynomial rates in the sample
size, meaning that reducing uncertainty requires exponentially large sample sizes.

We begin with the first approach and document a simple estimator under strong assumptions
on model specification. 

\begin{as}[Correct specification of outcome model]
\label{as:nonlinear_strong}
 $\E[Y_i \mid \mu_i, \sigma_i] = \rho_0 +
	\tau_0 f (\mu_i)$. Equivalently, $\E[\eta_i \mid \mu_i, \sigma_i] = 0 .$
\end{as}

\begin{prop}\label{prop:nonlinear_consistent}
Under \cref{as:nonlinear_strong,basicassp}, 
the population coefficient $\tau_0$ is equal to the regression coefficient of $Y_i$ on the
\emph{correctly specified empirical Bayes posterior mean} of $f(\mu_i)$ on $X_i, \sigma_i$
\[
\tau_0 = \frac{\cov(Y_i, \E[f(\mu_i) \mid X_i, \sigma_i])}{\var(\E[f(\mu_i) \mid
X_i, \sigma_i])}. \numberthis \label{eq:nonlinear_identi}
\]
\end{prop}
\begin{proof}
We can write \[
Y _i = \rho_0 + \tau_0 \E[f(\mu_i) \mid X_i, \sigma_i] + \tau_0 (f(\mu_i) - \E
[f(\mu_i) \mid X_i, \sigma_i]) + \eta_i. 
\]
Under \cref{as:nonlinear_strong}, $\E[\eta_i \mid X_i, \sigma_i] = 0$. By law of iterated
expectations, \[
\E[\tau_0 (f(\mu_i) - \E
[f(\mu_i) \mid X_i, \sigma_i]) \mid X_i, \sigma_i] = 0.
\]
Therefore, $\tau_0$ is equal to the population regression coefficient of $Y_i$ on
$\E_{G_i}[f(\mu_i) \mid X_i, \sigma_i]$.
\end{proof}

Operationizing \eqref{eq:nonlinear_identi} requires us to estimate $\E[f(\mu_i) \mid X_i,
\sigma_i]$ (importantly, distinct from $f(\E[\mu_i \mid X_i, \sigma_i])$). Traditional
empirical Bayes methods often specify a parametric model, e.g., $\mu_i
\mid
\sigma_i \sim \Norm (\mu,
\sigma_\mu^2)$. Regressing $Y_i$ on the estimated empirical Bayes posterior means under
such models can be interpreted as estimating \eqref{eq:nonlinear_identi}. Nevertheless,
the price of such an interpretation is the strong assumptions imposed:
\cref{as:nonlinear_strong} and well-specified parametric prior.

Without these strong assumptions, $\tau_0$ remains identified, as it is a function of the
joint distribution of $(Y_i, \mu_i)$. This joint distribution is identified from the joint
distribution of $(Y_i, X_i, \sigma_i)$ via deconvolution.  However, as is typical with
deconvolution, the problem of estimating $\tau_0$ becomes prohibitively difficult without
making parametric restrictions \citep{fan1993nonparametric}. We illustrate this point with
a minimax lower bound for $\tau_0$. 

Let us first define the class of distributions the lower bound is over. Fix some function
$f: [-1,1] \to \R$ such that it is bounded ($\norm{f}_\infty < \infty$) and nonconstant in
the sense that $V(f) \equiv \var_{U \sim \Unif[-1,1]}(f(U)) > 0$. Let $\mathcal Q_0$
collect all distributions $Q_0$ for $(Y_i ,\mu_i, X_i,
\sigma_i)$ satisfying \cref{basicassp} where (a) $\mu_i$ and $Y_i$ are supported within
the interval $ [-1,1]$, (b) $\var_{Q_0}(f(\mu_i)) > \frac{1}{2} V(f)$ to avoid degeneracy,
and (c) $\sigma_i \le 1$.\footnote{The restrictions made on $\mathcal Q_0$ is for
convenience. Note that the minimax rate over a larger set of distributions is necessarily
bounded below by the minimax rate over $\mathcal Q_0$. } Let $\mathcal Q = \br{P^\obs
(P_0) : P_0 \in \mathcal Q_0}$. Relative to $f$, let $\tau_0(Q_0)
=
\frac{\cov_{Q_0}(Y, f(\mu))}{\var_{Q_0}(f(\mu))}$. The \emph{minimax risk} of estimating $\tau_0
(Q_0)$ is the worst-case squared error of a given estimator $T$ over $\mathcal Q_0$, \[
R_n(\mathcal Q_0, f) = \inf_T \sup_{Q_0 \in \mathcal Q_0} \E_{Q_0}\bk{
	\br{T(Y_{1:n}, X_{1:n}, \sigma_{1:n}) - \tau_0(Q_0)}^2
},
\]
optimized over choices of all estimators. 

The minimax rate measures the difficulty of estimating $\tau_0(Q_0)$ over $Q_0$. One way
to interpret $R_n$ is as the value of a zero-sum game where an analyst moves first and
chooses and estimator $T$, and an adversary moves second and chooses a distribution $Q_0
\in \mathcal Q_0$. If different distributions $Q_0, Q_1 \in \mathcal Q_0$ produce very
different $\tau_0$ but very similar data, then the analyst would suffer large losses as
they would be unable to distinguish these scenarios. 

For well-behaved estimands (e.g. when $f(\mu) =
\mu$), $R_n = O(1/\sqrt{n})$ contracts at the familiar parametric rate.\footnote{To wit,
 the restrictions on $\mathcal Q_0$ implies that $\tau_0(Q_0)$ is bounded by some $M >
 0$. Thus, we can consider the estimator $T = \max(\min(\hat\beta, M), -M)$. The
 truncation at $M$ is so that expectations always exist. } In sharp contrast, when $f
 (\cdot)$ is not an analytic function, the minimax rate vanishes slower than any
 polynomial in $n$. To reduce the uncertainty in our estimates proportionally by $t \in 
 (0,1)$, we
 therefore would require sample sizes exponential in $1/t$. 

\begin{restatable}{theorem}{thmminimax}
\label{thm:minimax}
Under this setup, if $f$ is not an analytic function,\footnote{A real-valued function $f:
[-1,1] \to \R$ is \emph{analytic} if there exists an extension of $f$ on an open subset of
$\C$, $\tilde f : U \to \C$ where $U\subset \C$ is open, such that $\tilde f = f$ on $
[-1,1]$ and $\tilde f$ is complex analytic (i.e. complex differentiable). } then for any
$\alpha > 0$, $R_n$ contracts at a rate slower than $n^{-\alpha}$: $
\limsup_{n\to\infty} n^{\alpha} R_n(\mathcal Q_0, f) = \infty. 
$
\end{restatable}

Analytic functions are infinitely differentiable and admit Taylor expansions everywhere.
However, many---if not most---transformations of interest are not analytic as they are
typically non-smooth. For these functions, \cref{thm:minimax} is then a negative result,
in the sense that the regression coefficient $\tau_0$ associated with them is
fundamentally difficult to estimate without assuming further structure.

The proof to \cref{thm:minimax} uses Le Cam's two-point method by constructing $Q_1, Q_0
\in \mathcal Q_0$ that generate similar data but very different $\tau_0(Q)$. A key step of
 the proof constructs $Q_1, Q_0$ so as to reduce the problem of estimating the regression
 coefficient $\tau_0$ to the problem of estimating a mean $\E_Q[f(\mu)]$. Minimax rates
 for the latter problem then follow from the techniques developed by \citet
 {cai2011testing} and presented in \citet{wu2020polynomial}.\footnote{\citet
 {cai2011testing} specifically study estimating $\E|\mu|$. Their proof technique extends
 to $\E[f(\mu)]$ by applying some results in approximation theory known as Bernstein's
 and Jackson's theorems. We have not seen these results stated in the statistical
 literature. Our proof of \cref{thm:minimax} makes it precise. } The proof to \cref
 {thm:minimax} similarly applies to regression problems where $f(\mu)$ appears on the
 left-hand side. That is, the minimax rate for the population regression coefficient
 $\tau_0$ in the infeasible specification\[ f(\mu_i) = \rho_0 + \tau_0 W_i + \eta_i
\]
suffers from similar subpolynomial rates if $f$ is non-analytic.

\begin{figure}[!th]
  \centering
  \includegraphics[width=\textwidth]{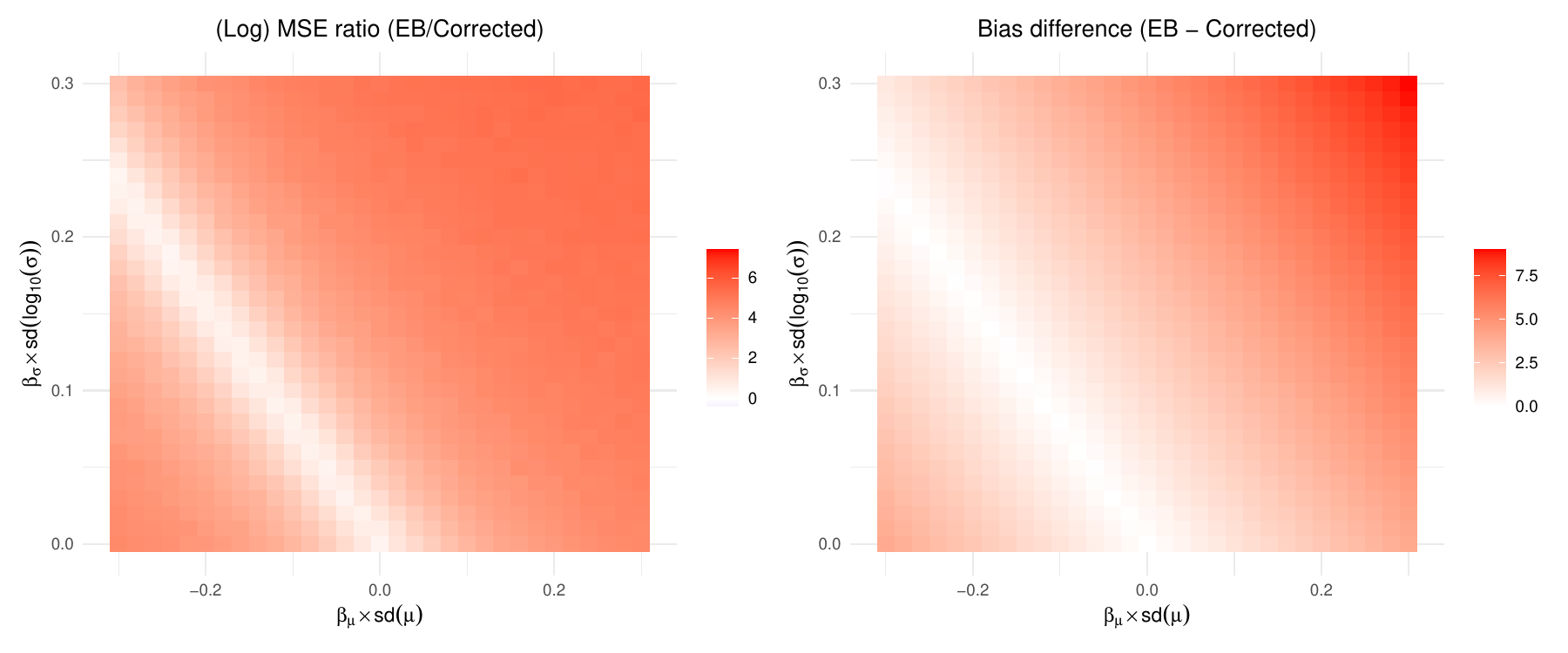}
     \caption{Simulation results for linear-in-$\mu$
      regressions.}
    \caption*{\footnotesize{\emph{Notes}: The left figure shows a heat map of the (natural) log of the
      ratio of the mean squared error (MSE) of $\hat\beta$ and
      $\tilde\beta$, over 1,000 simulation draws. The right figure shows a similar heat map, but for the
      difference in the bias of the two estimators.}}
    \label{fig:sim_linear}
  \end{figure}
  \section{Simulation}
\label{sec:simulation}
We follow \cite{chen2023empirical} and calibrate our data generating process (DGP) to the
Opportunity Atlas (\citealp{chetty2020opportunity}), which provides (unshrunk)
economic mobility estimates $X_{i}$ and their standard errors $\sigma_{i}$. One
measure of economic mobility $\mu_{i}$ of tract $i$ they consider is the
probability that a Black individual becomes relatively high-income (i.e., has
family income in the top 20 percentiles nationally) after growing up relatively
poor in tract $i$. (i.e., with parents at the 25th percentile
nationally). We have 10,058 tracts in the data.

Taking this mobility measure as our measure of interest, we estimate the conditional mean
function $ \E [\mu_{i} \mid  \sigma_{i} ] = \E[X_{i}  \mid \sigma_{i} ]$ and the conditional
variance function $\var(\mu_{i} \mid\sigma_{i}) = \var(X_{i}\mid \sigma_{i}) -\sigma_{i}^{2}$
via local linear regression implemented by \cite{calonico2019nprobust}; denote the
estimates of the conditional mean and conditional variance functions as $\hat{m}(\cdot)$
and $\hat{s}^{2}(\cdot)$, respectively. We use these estimates to draw $\mu_{i}$'s from
the distribution $\mu_{i} \mid \sigma_{i} \sim \Norm(\hat{m}(\sigma_{i}),
\hat{s}^{2}(\sigma_{i}))$. The outcome variable $Y_{i}$ for the regression of
interest is generated by
\[
  Y_{i} = \beta_{\mu} \mu_{i} + \beta_{\sigma} \log_{10} \sigma_{i} + u_{i},
\]
where $u_{i} \sim \Norm(0, 1)$. Under this DGP, the linear projection coefficient---in a
hypothetical regression of $Y_i$ on $\mu_i$---we wish to estimate is given by
\[\beta_{0} = \beta_{\mu} + \frac{\cov(\mu_{i}, \log_{10} \sigma_{i})}{\var(\mu_{i})}
  \beta_{\sigma}.\] We vary $(\beta_{\mu},\beta_{\sigma})$\footnote{Specifically,
  we vary $(\beta_{\mu},\beta_{\sigma})$ so that $\var(\mu)^{1/2}\beta_{\mu}$
  ranges from $-.3$ to $.3$ and $\var(\log_{10}(\sigma))^{1/2}\beta_{\sigma}$
  from $0$ to $.3$. Combined with the fact that we take $\var(u_{i}) = 1$, this
  results in a DGP where the regressors have realistic explanatory power.}
and compare the performance of our proposed estimator $\hat{\beta}$, given in
\eqref{eq:inf_est}, with the regression-on-shrinkage estimator $\tilde{\beta}$,
given in \eqref{eq:reg_shr_est}.

Figure \ref{fig:sim_linear} summarizes the results of our simulation exercise in
this setting by comparing the MSE and bias of the two estimators across the
different DGPs. As expected, the MSE of $\hat{\beta}$ is smaller than the
$\tilde{\beta}$ in almost all specifications, and this improvement is
substantial in a wide range of simulations. For example, when
$\var(\mu)^{1/2}\beta_{\mu} = .2$ and
$\var(\log_{10}(\sigma))^{1/2}\beta_{\sigma} = .05$, which is a setting where,
roughly speaking, $\log_{10}(\sigma)$ has a low explanatory power compared to
$\mu$, the log MSE ratio is around $4.36$, which indicates that the MSE of
$\tilde{\beta}$ is about $78$ times greater than that of
$\hat{\beta}$. Similarly, the bias
plot shows that $\hat{\beta}$ having a smaller bias than
$\tilde{\beta}$ across all DGPs, which is expected given that $\hat{\beta}$ is
unbiased across all DGPs we consider.

For the simulations regarding nonlinear
transformations of $\mu_{i}$, we use a simpler DGP for $\mu_{i}$. Specifically, we take an estimate of the unconditional mean and variance implied by
the previous DGP,
$\hat{m} = n^{-1}\sum_{i=1}^{n}\hat{m}(\sigma_{i})$ and $\hat{s}^{2} =
n^{-1}\sum_{i=1}^{n}\hat{s}^{2}(\sigma_{i}) +
n^{-1}\sum_{i=1}^{n}(\hat{m}(\sigma_{i}) - \hat{m})^{2}$ and draw $\mu_{i} \sim
G$, where $G$ is the distribution function of $\Norm(\hat{m}, \hat{s}^{2})$. Then, we generate the outcome variable as
\[
Y_{i} = \tau_{0}\one(\mu_i > \mu_0) + u_{i},
\]
where $u_{i} \sim N(0, 1)$ is independent of $(\mu_{i}, \sigma_{i})$, and
$\mu_{0}$ is some fixed threshold assumed to be known. We vary $\tau_{0}$ and
$\mu_{0}$ to learn about the performance of the different estimators we describe
below.

\begin{figure}[tbp]
  \centering
  \begin{subfigure}{.49\textwidth}
    \centering
    \includegraphics[width=\textwidth]{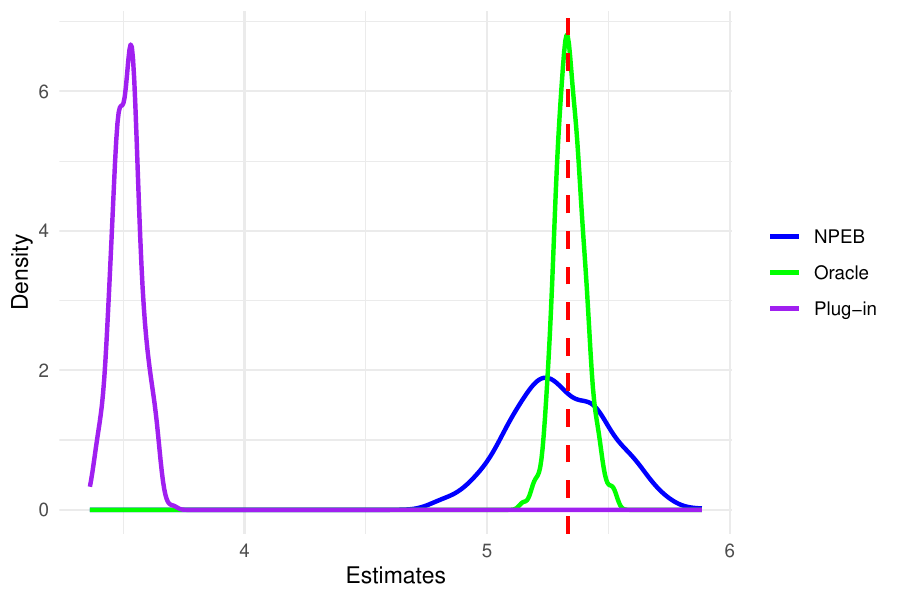}
    \caption{$\mu_{0}=$.75-quantile of $G$}
    \label{fig:fig1}
  \end{subfigure}
  \begin{subfigure}{.49\textwidth}
    \centering
    \includegraphics[width=\textwidth]{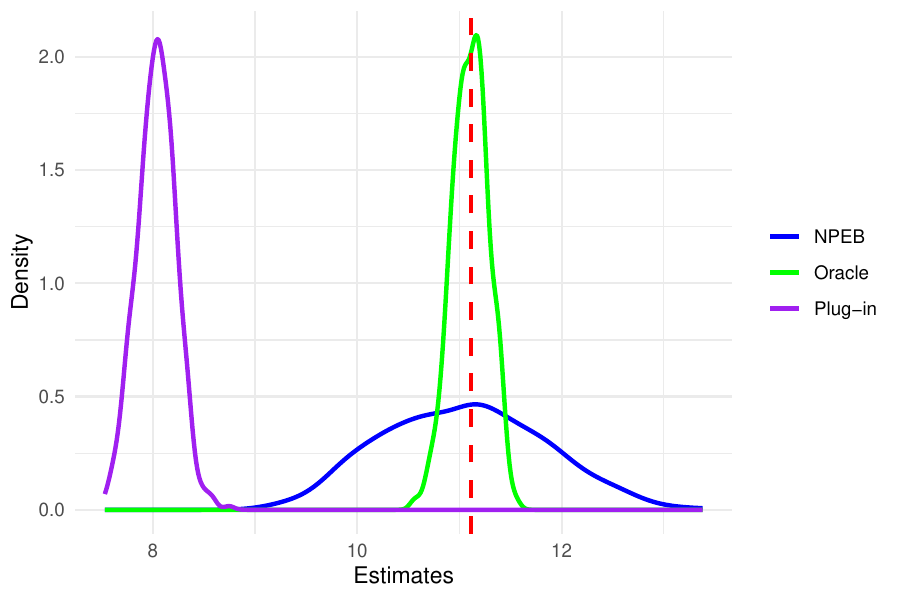}
    \caption{$\mu_{0}=$.9-quantile of $G$}
    \label{fig:fig2}
  \end{subfigure}
  \caption{Simulation results for regressions on nonlinear transformations of
    $\mu$.}
  \caption*{\footnotesize{\emph{Notes}: The figures show density plots of three different estimates across 500
    simulation draws. The left figure shows results for a DGP where the
    threshold level $\mu_{0}$ is set at the .75-quantile of $G$. The right
    figure shows the same for a DGP where the threshold level $\mu_{0}$ is set
    at the .90-quantile of $G$. The red dashed line shows the true parameter value.}}
  \label{fig:sim_nonlinear}
\end{figure}

We consider three distinct estimators. First is the oracle estimator
$\hat{\tau}_{0}^{\mathrm{oracle}}$ obtained by regressing $Y_{i}$ on the true
$\E_{G}[\one(\mu_i > \mu_0) \mid X_{i}, \sigma_{i}]$ using knowledge of the true $G$. By
Proposition \ref{prop:nonlinear_consistent}, $\hat{\tau}_{0}^{\mathrm{oracle}}$ is
consistent at the usual $n^{-1/2}$-rate. The second estimator we consider is the
nonparametric empirical Bayes (NPEB) estimator $\hat{\tau}^{\mathrm{NPEB}}_{0}$ obtained
by regressing $Y_{i}$ on $\E_{\hat{G}}[\one(\mu_i > \mu_0) \mid X_{i}, \sigma_{i}]$, where
$\hat{G}$ is obtained by using nonparametric maximum likelihood as in
\cite{gilraine2020new}. This estimator is consistent, but potentially at a slower rate, as
implied by \cref{thm:minimax}. Finally, we also consider a plug-in estimator
$\hat{\tau}_{0}^{\text{plug-in}}$ which is obtained by regressing $Y_{i}$ on
$\one(\E_{G}[\mu_{i} \mid X_{i}, \sigma_{i}] > \mu_{0})$.\footnote{We use an
``infeasible'' plug-in estimator in the sense that we take the true $G$ to calculate
$\E_{G}[\mu_{i} \mid X_{i}, \sigma_{i}]$. Hence, the results should be considered an upper
bound on how well a plug-in rule can do.} This estimator mimics the rather common
empirical practice of plugging in the shrunk estimates $\hat{\mu}_{i}$ for $\mu_{i}$ in
various downstream analyses.

Figure \ref{fig:sim_nonlinear} shows the simulation results for two DGPs where
we set $\tau_{0}$ so that $\var(\one(\mu_i > \mu_0))^{1/2}\tau_{0} = 1$. We
consider two threshold levels for $\mu_{0}$, the .75-quantile and .90-quantile
of $G$. The plug-in estimator $\hat{\tau}_{0}^{\text{plug-in}}$ is
substantially biased, as expected since $\hat{\tau}_{0}^{\text{plug-in}}$
plugging the posterior of $\mu_{i}$ into nonlinear transformation does not in
general correct the measurement error. On the other hand, we see that
$\hat{\tau}_{0}^{\mathrm{oracle}}$ is unbiased, as expected by Proposition
\ref{prop:nonlinear_consistent}. Note that $\hat{\tau}_{0}^{\mathrm{oracle}}$
shows higher accuracy when $\mu_{0}$ is set at the .75-quantile of $G$, because
this results in regressors with higher variance. Finally, while the estimates
generated by $\hat{\tau}_{0}^{\mathrm{NPEB}}$ are centered around the true
parameter value, the distribution is quite dispersed and substantially noisier
than $\hat{\tau}_{0}^{\mathrm{oracle}}$. Given that the sample size is
moderately large at 10,058, this demonstrates the slow convergence rate
predicted by Theorem \ref{thm:minimax}. The simulation results for other
specifications of $\tau_{0}$ and $\mu_{0}$ are qualitatively similar.

\section{Empirical applications}
\label{sec:empirical}

\subsection{Impact of examiner on patent outcome}
We revisit how patent examiners leniency (i.e., leniency in patent crafting) predicts the
outcomes of the patents they decide to grant \citep{feng2020crafting}. As background, for
each patent application, an examiner from a specific art unit is randomly assigned to
assess the application and determine its merit for patent granting. There is a significant
interaction between patent applicants and examiners during the revision of claims until
the patent is either granted or the application is withdrawn. All other factors being
equal, a stricter examiner may raise numerous questions regarding prior art, appropriate
citations, and required clarifications before granting the patent, while a lenient
examiner may provide minimal feedback for similar patents. This crafting process could
influence the quality and clarity of the patent, should it be granted, which in turn could
affect its market value, litigation propensity, and citations.

In particular, we consider  the following infeasible regression:
\[
Y_{j}= \beta_0 \mu_{i(j)} + a_{ut(j)} + \epsilon_j
\]
where $j$ indexes the patent, $i$ the examiner, $u$ the art unit, and $t$ the filing year.
The parameter $\beta_0$ represents the effect of examiner $i$'s leniency during the
crafting process, denoted as $\mu_{i(j)}$, on the patent outcome $Y_j$. Following
\cite{feng2020crafting}, we use two measures of leniency. The first measure accounts for
the average percentage change in the number of words per claim between the application and
post-grant stages. The second measure records the percentage change in the total number of
claims. These measures are constructed using the pre- and post-grant publication
data\footnote{These data are publically available at
\url{https://patentsview.org/download/pg_claims} and
\url{https://patentsview.org/download/claims}.}.

An examiner who demands an increase in the length of the claims---typically to clarify
distinctions from existing grants and enhance precision---is considered more careful and
stringent. Similarly, an examiner who requests a reduction in the number of claims to
limit the scope of the patent is also regarded as more stringent. As expected, these
measures are, at best, noisy estimates of the true leniency $\mu_i$.

We focus on the citation outcomes of patents within three years of grant. Additionally, we
report results on the probability of purchase and litigation by a Patent Assertion Entity
(PAE). In this case, the outcome variable takes a value of 1 if the granted patent is
purchased by a PAE and is also litigated in court for infringement. We obtain the set of
patents litigated by PAE firms from the Stanford NPE Litigation Database, which
categorizes assertors into several categories. We classify acquired patents, failed
startups, individual-inventor-stated companies, and individuals as Patent Assertion
Entities (PAEs).  As documented in \cite{feng2020crafting}, these entities acquire patents
from third parties and generate revenue by asserting them against alleged infringers,
commonly known as patent trolls. We also include additional cases from the Unified Patent
and Lex Machina databases\footnote{We thank Dr. Tommaso Alba for providing these data.}. The
citation data is publicly available from the USPTO and includes all citations received by
a patent since its grant, updated annually.

The discussion below excludes the art unit and year fixed effects. In practice, we first
purge these fixed effects from all other variables as discussed in Section \ref{sub:implement}.\footnote{We consider only those
examiners who serve in a single art unit and have granted at least 10 patents, leaving us
with 4,615 examiners.} Let the measure of leniency for each patent $j$ reviewed by examiner $i$ be denoted by $E_{ij}$ and $N_i$
represents the number of patents granted by examiner $i$. Then $Y_i = \frac{1}{N_i}
\sum_{j: i(j) = i}Y_{ij}$ is the examiner-level mean outcome\footnote{Here we index the
patent level outcome by $Y_{ij}$ rather than $Y_j$ to clarify that the patent level
variables are more granular than the examiner level.}, and $E_i= \frac{1}{N_i} \sum_{j:
i(j) = i} E_{ij}$ is the examiner's leniency, for which
$
	E_i \mid \mu_i \sim N(\mu_i, \sigma_i^2).
$
where $\sigma_i^2$ is proportional to the inverse of $N_i$. 	
	
Replacing $\mu_{i(j)}$ with $E_{i(j)}$ introduces attenuation bias. Accounting for the fact that we aggregate the patent level regression to the examiner level, the population least-square objective is 
\[
\mathbb{E} \Big[ \sum_{j: i(j) = i} (Y_{ij} - \delta_0 - \beta_0 \mu_i)^2 \Big ] = \mathbb{E}[N_i (Y_i - \delta_0 - \beta_0 \mu_i)^2] + \text{constant}
\]
Therefore, the population coefficient $\beta_0$ is the weighted least square coefficient at the examiner level, with weights proportional to $N_i$. The natural sample analogue, correcting for measurement error, leads to: 
\[
\hat \beta = \frac{\sum_i W_i Y_i E_i - (\sum_i W_i E_i)(\sum_i W_i Y_i)}{\sum_i W_i E_i^2 - (\sum_i W_i E_i)^2 - \sum_i W_i \sigma_i^2}
\quad W_i = N_i/\sum_i N_i.\]

An additional concern may be that the left-hand side variable is also a measurement of a latent effect and contains measurement error. In this case, it is possible that the measurement error in both $Y_{ij}$ and $E_{ij}$ are correlated. To model this, we consider the following model 
\[
\begin{pmatrix} Y_{ij} \\ E_{ij} \end{pmatrix} \sim N \left  ( \begin{pmatrix} \rho_i \\ \mu_i \end{pmatrix}, \Sigma \right)\quad \text{and } \quad \begin{pmatrix} Y_{i} \\ E_{i} \end{pmatrix} \sim N \left  ( \begin{pmatrix} \rho_i \\ \mu_i \end{pmatrix}, \Sigma_i \right) 
\]
where the second display is the model for the examiner-specific sample averages $(Y_i,
E_i)$ and $\Sigma_i = \Sigma/N_i$. The coefficient $\beta_0$ can be related to the
population coefficient of a regression of $Y_i$ on $E_i$ by
$
\beta = \frac{\cov(Y_i, E_i) - \Sigma_{i, 12}}{\var(E_i) - \Sigma_{i, 22}}.
$
In the presence of heterogeneous $N_i$, the natural weighted least square estimator with the correction due to measurement error on both sides yields 
\[
\check{\beta} = \frac{\sum_i W_i Y_i E_i  - (\sum_iW_i E_i) (\sum_i W_i Y_i) - \sum_i W_i \hat \Sigma_{i,12}}{\sum_i W_i E_i^2 - (\sum_i W_iE_i)^2 - \sum_i W_i \sigma_i^2}
\]
where $\hat \Sigma_{i,12}$ is estimated based on patent level data $(Y_{ij}, E_{ij})$ and normalized by $N_i$. 

In contrast to the two proposed estimators above, the common regress-on-shrinkage approach constructs the posterior mean as:
\[
\bar E_i = \hat \mu_0 + \frac{\hat \sigma_e^2}{\sigma_i^2 + \hat \sigma_e^2} (E_i - \hat \mu_0)
\]
where $\hat \mu_0$ and $\hat \sigma_e^2$ are often estimated via moments of $E_i$. The resulting weighted least square estimator is: 
\[
\tilde \beta = \frac{\sum_i W_i Y_i \bar E_i - (\sum_i W_i \bar E_i)(\sum_i W_i Y_i)}{\sum_i W_i \bar E_i^2 - (\sum_i W_i \bar E_i)^2}.
\]

\Cref{tab: result} reports the estimates for $\beta_0$ using different methods. The column
with label FE represents the OLS estimation without correcting for measurement error in
the leniency measures. This estimator biases towards zero due to un-corrected measurement error. The next column reports the regress-on-shrinkage estimator.  
The last two columns (labeled as FE-Corrected and FE-Corrected-Twoside) report estimates
based on the measurement error correction. FE-Corrected only considers measurement error
in leniency measures while FE-Corrected-Twoside accounts for measurement error from both the
left and right hand side variables. The standard errors in the bracket for our proposed
method are obtained via bootstrap with 999 bootstrap repetitions.

\begin{table}[h!] 
	\centering
	
	\caption{Patent Level Outcome Regression }
    \label{tab: result} 
	\begin{tabular}{@{\extracolsep{0pt}}ll ccc} 
			\\[-1.8ex]\hline 
		\hline \\[-1.8ex]
		& FE & Shrinkage & FE-Corrected & FE-Corrected-Twoside \\ 
		\hline
		\hline 
	&	\multicolumn{3}{r}{Citation within 3 year}\\
		\hline 
	word change & $-3.307^{***}$ & $-3.950^{***}$ & $-3.902^{***}$ & $-3.786^{***}$ \\ 
	 & (0.536) & (0.614) & (0.805) & (0.841) \\ 
claims change & $2.924^{***}$ & $5.705^{***}$ & $4.307^{*}$& $4.275^{*}$ \\ 
 & (0.780) & (1.110) & (2.203) & (2.236) \\ 
			\hline 
				&	\multicolumn{3}{r}{PAE Litigated}\\
				\hline 
				word change & $-0.0034^{***}$ & $-0.0039^{***}$ & $-0.0040^{***}$ & $-0.0040^{***}$ \\ 
				 & (0.0007) & (0.0008) & (0.0009) & (0.0009) \\ 
				claims change & $0.0027^{**}$ & $0.0043^{**}$ & $0.0039^{**}$ & $0.0037^{**}$ \\ 
				 & (0.0010) & (0.0014) & (0.0016) & (0.0017) \\ 
				\hline 
					\hline \\[-1.8ex] 
			\textit{Note:}  &\multicolumn{4}{r}{$^{*}$p$<$0.1; $^{**}$p$<$0.05; $^{***}$p$<$0.01.} 
	\end{tabular}
\end{table}

For the PAE litigation outcome, there is little difference between the shrinkage method
and the direct correction we proposed. However, for the citation outcome, the difference
is noticeable. As reported in \cref{tab: sigma}, individual variances $\sigma_i$ appear to
have a significant impact on the citation outcome, but not on the PAE litigation outcome.
This suggests some evidence of a violation of the assumptions necessary to ensure the
consistency of the regress-on-shrinkage estimator.

\begin{table}[!htbp] \centering
		\caption{Regression on logarithm of $\sigma_i$.} \label{tab: sigma} 
	\begin{tabular}{@{\extracolsep{5pt}}lD{.}{.}{-3} D{.}{.}{-3} } 
		\\[-1.8ex]\hline 
		\hline \\[-1.8ex] 
		& \multicolumn{2}{c}{\textit{Dependent variable:}} \\ 
		\cline{2-3} 
		\\[-1.8ex] & \multicolumn{1}{c}{Citation within 3 years} & \multicolumn{1}{c}{PAE Litigated} \\ 
		\hline \\[-1.8ex] 
		$\log_{10}(\sigma_i) $& -0.749^{***} & -0.0002 \\ 
		& (0.141) & (0.0002) \\ 
		Constant & -1.346^{***} & -0.0004 \\ 
		& (0.256) & (0.0003) \\ 
		\hline 
		\hline \\[-1.8ex] 
		\textit{Note:}  & \multicolumn{2}{r}{$^{*}$p$<$0.1; $^{**}$p$<$0.05; $^{***}$p$<$0.01} \\ 
	\end{tabular} 
\end{table} 

\subsection{Teacher value-added in a low-income country} 

We revisit 
\citet{bau2020teacher}, who compute teacher-value added in a development economics
 context. In their Table 8, they consider whether value-added predicts log salary of
 teachers, for teachers in the public sector and in the private sector. \citet
 {bau2020teacher} find that value added does not predict salary for teachers in the
 public sector but do for those in the private sector.

\begin{figure}[htb]
  \centering
  \includegraphics{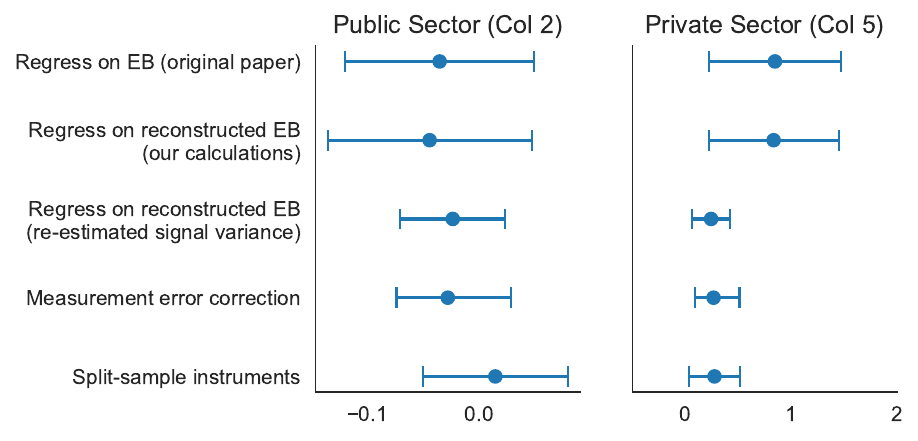}

\begin{proof}[Notes]
We consider the specifications (2) and (5) in Table 8 in \citet{bau2020teacher}.  The
  first row uses EB posterior mean estimates from \citet{bau2020teacher} and estimates
  \eqref{eq:reg_shr_est}, adjusting for additional covariates. The second row reconstructs
   these estimates---different from the original paper due to small sample differences.
   Owing to the inconsistency between $\hat\sigma_\mu^2$, $\sigma_i^2$, and the sample
   variance of the value-added estimates, the third row reestimates $\sigma_\mu^2$
   through an analogue estimator and recomputes the empirical Bayes estimates. Using the
   recomputed $\sigma_\mu^2$, the fourth row implements the standard measurement error
   correction. Finally, the fifth row constructs two independent noisy estimates by
   adding and subtracting Gaussian noises, and uses one noisy estimate to instrument for
   the other.
\end{proof}

  \caption{Robust estimators for the predictive relationship between teacher value-added
  and log salary (Table 8, \citet{bau2020teacher})}. 
  \label{fig:baudas}
\end{figure}

For generating shrunken teacher-value
added estimates, \citet{bau2020teacher} do not directly compute from
\eqref{eq:par_emp_bayes} given a set of preliminary noisy estimates. Rather, both the
 signal variance $\hat\sigma_\mu^2$ and the idiosyncratic variances $\sigma_i^2$ are
 estimated jointly through a sample-splitting procedure. A somewhat unusual feature of
 this computation is that the estimated total variance of teacher estimates does not
 equal the sum of $\hat\sigma_\mu^2$ and $\sigma_i^2$. In what follows, we take $\E_n
 \sigma_i^2$ from this procedure as the sampling variance of the noisy estimates, and we
  compute $\hat\sigma_\mu^2 = \var_n(X_i) - \E_n\sigma_i^2$ as an analogue estimator.
  See \cref{sec:details_for_the_empirical_application} for details.

We show alternative estimates in \cref{fig:baudas}. The differences in the estimates are
broadly explained by different methods for estimating $\sigma_\mu^2$. In this empirical example, our proposed estimator yields results similar to those of the conventional method, reaffirming the authors' original findings under weaker conditions. This is consistent with the fact that the shrinkage
factors do not appear to be strongly correlated with $\mu_i$.

\section{Conclusion} This paper critically examines the widespread practice of using the
 empirical Bayes shrinkage estimator to correct for measurement errors in regression
 models incorporating individual latent effects. Our analysis reveals that this approach
 only provides reliable estimates for the regression coefficients under unnecessarily
 stringent conditions. We demonstrate that the classical correction, which we advocate,
 holds under weaker assumptions and cannot be asymptotically improved when the latent
 effect enters the regression model linearly. In cases where the latent effect enters
 non-linearly, empirical Bayes shrinkage leads to slower minimax estimation rates. These
 findings underscore the limitations of using the regress-on-shrinkage estimator as a
 method for correcting measurement error in regression models.

\newpage
	
	\bibliographystyle{chicago}
 \bibliography{main.bib}

\begin{thebibliography}{44}
\providecommand{\natexlab}[1]{#1}

\bibitem[{Angelova \textit{et~al.}(2023)Angelova, Dobbie and
  Yang}]{angelova2023algorithmic}
\textsc{Angelova, V.}, \textsc{Dobbie, W.~S.} and \textsc{Yang, C.} (2023).
  \textit{Algorithmic recommendations and human discretion}. Tech. rep.,
  National Bureau of Economic Research.

\bibitem[{Angrist \textit{et~al.}(2023)Angrist, Hull and
  Walters}]{angrist2023methods}
\textsc{Angrist, J.}, \textsc{Hull, P.} and \textsc{Walters, C.} (2023).
  Methods for measuring school effectiveness. \textit{Handbook of the Economics
  of Education}, \textbf{7}, 1--60.

\bibitem[{Battaglia \textit{et~al.}(2024)Battaglia, Christensen, Hansen and
  Sacher}]{battaglia2024inference}
\textsc{Battaglia, L.}, \textsc{Christensen, T.}, \textsc{Hansen, S.} and
  \textsc{Sacher, S.} (2024). Inference for regression with variables generated
  by ai or machine learning.

\bibitem[{Bau and Das(2020)}]{bau2020teacher}
\textsc{Bau, N.} and \textsc{Das, J.} (2020). Teacher value added in a
  low-income country. \textit{American Economic Journal: Economic Policy},
  \textbf{12}~(1), 62--96.

\bibitem[{Bickel and Ritov(1987)}]{bickel1987efficient}
\textsc{Bickel, P.} and \textsc{Ritov, Y.} (1987). Efficient estimation in the
  errors in variables model. \textit{The Annals of Statistics},
  \textbf{15}~(2), 513--540.

\bibitem[{Bruhn \textit{et~al.}(2022)Bruhn, Imberman and
  Winters}]{bruhn2022regulatory}
\textsc{Bruhn, J.}, \textsc{Imberman, S.} and \textsc{Winters, M.} (2022).
  Regulatory arbitrage in teacher hiring and retention: {{Evidence}} from
  {{Massachusetts Charter Schools}}. \textit{Journal of Public Economics},
  \textbf{215}, 104750.

\bibitem[{Cai and Low(2011)}]{cai2011testing}
\textsc{Cai, T.~T.} and \textsc{Low, M.~G.} (2011). Testing composite
  hypotheses, hermite polynomials and optimal estimation of a nonsmooth
  functional. \textit{The Annals of Statistics}, pp. 1012--1041.

\bibitem[{Calonico \textit{et~al.}(2019)Calonico, Cattaneo and
  Farrell}]{calonico2019nprobust}
\textsc{Calonico, S.}, \textsc{Cattaneo, M.~D.} and \textsc{Farrell, M.~H.}
  (2019). nprobust: Nonparametric kernel-based estimation and robust
  bias-corrected inference. \textit{Journal of Statistical Software},
  \textbf{91}, 1--33.

\bibitem[{Chandra \textit{et~al.}(2016)Chandra, Finkelstein, Sacarny and
  Syverson}]{chandra2016health}
\textsc{Chandra, A.}, \textsc{Finkelstein, A.}, \textsc{Sacarny, A.} and
  \textsc{Syverson, C.} (2016). Health care exceptionalism? performance and
  allocation in the us health care sector. \textit{American Economic Review},
  \textbf{106}~(8), 2110--2144.

\bibitem[{Chen(2025)}]{chen2023empirical}
\textsc{Chen, J.} (2025). Empirical bayes when estimation precision predicts
  parameters. \textit{arXiv preprint arXiv:2212.14444}.

\bibitem[{Chen and Santos(2018)}]{chen2018overidentification}
\textsc{Chen, X.} and \textsc{Santos, A.} (2018). Overidentification in regular
  models. \textit{Econometrica}, \textbf{86}~(5), 1771--1817.

\bibitem[{Chetty \textit{et~al.}(2020)Chetty, Friedman, Hendren, Jones and
  Porter}]{chetty2020opportunity}
\textsc{Chetty, R.}, \textsc{Friedman, J.~N.}, \textsc{Hendren, N.},
  \textsc{Jones, M.~R.} and \textsc{Porter, S.} (2020). \textit{The Opportunity
  Atlas Mapping the Childhood Roots of Social Mobility}. Tech. rep.

\bibitem[{Chetty \textit{et~al.}(2014{\natexlab{a}})Chetty, Friedman and
  Rockoff}]{chetty2014measuringa}
\textsc{---}, \textsc{---} and \textsc{Rockoff, J.~E.} (2014{\natexlab{a}}).
  Measuring the impacts of teachers {I}: {E}valuating bias in teacher
  value-added estimates. \textit{American Economic Review}, \textbf{104}~(9),
  2593--2632.

\bibitem[{Chetty \textit{et~al.}(2014{\natexlab{b}})Chetty, Friedman and
  Rockoff}]{chetty2014measuringb}
\textsc{---}, \textsc{---} and \textsc{---} (2014{\natexlab{b}}). Measuring the
  impacts of teachers {II}: {T}eacher value-added and student outcomes in
  adulthood. \textit{American Economic Review}, \textbf{104}~(9), 2633--79.

\bibitem[{Chetty and Hendren(2018)}]{chetty2018impacts}
\textsc{---} and \textsc{Hendren, N.} (2018). The impacts of neighborhoods on
  intergenerational mobility {II}: {C}ounty-level estimates. \textit{Quarterly
  Journal of Economics}, \textbf{133}~(3), 1163--1228.

\bibitem[{de~Chaisemartin and Deeb(2024)}]{de2024estimating}
\textsc{de~Chaisemartin, C.} and \textsc{Deeb, A.} (2024). Estimating
  treatment-effect heterogeneity across sites in multi-site randomized
  experiments with imperfect compliance. \textit{arXiv e-prints}, pp.
  arXiv--2405.

\bibitem[{Deaton(1985)}]{deaton1985panel}
\textsc{Deaton, A.} (1985). Panel data from time series of cross-sections.
  \textit{Journal of econometrics}, \textbf{30}~(1-2), 109--126.

\bibitem[{Deeb(2021)}]{deeb2021framework}
\textsc{Deeb, A.} (2021). A framework for using value-added in regressions.
  \textit{arXiv preprint arXiv:2109.01741}.

\bibitem[{Devereux(2007)}]{devereux2007improved}
\textsc{Devereux, P.~J.} (2007). Improved errors-in-variables estimators for
  grouped data. \textit{Journal of Business \& Economic Statistics},
  \textbf{25}~(3), 278--287.

\bibitem[{DeVore and Lorentz(1993)}]{devore1993constructive}
\textsc{DeVore, R.~A.} and \textsc{Lorentz, G.~G.} (1993). \textit{Constructive
  approximation}, vol. 303. Springer Science \& Business Media.

\bibitem[{Fan and Truong(1993)}]{fan1993nonparametric}
\textsc{Fan, J.} and \textsc{Truong, Y.~K.} (1993). Nonparametric regression
  with errors in variables. \textit{The Annals of Statistics}, pp. 1900--1925.

\bibitem[{Feng and Jaravel(2020)}]{feng2020crafting}
\textsc{Feng, J.} and \textsc{Jaravel, X.} (2020). Crafting intellectual
  property rights: Implications for patent assertion entities, litigation, and
  innovation. \textit{American Economic Journal: Applied Economics},
  \textbf{12}~(1), 140--181.

\bibitem[{Fenizia(2022)}]{fenizia2022managers}
\textsc{Fenizia, A.} (2022). Managers and productivity in the public sector.
  \textit{Econometrica}, \textbf{90}~(3), 1063--1084.

\bibitem[{Fuller(1987)}]{fuller2009measurement}
\textsc{Fuller, W.~A.} (1987). \textit{Measurement error models}. John Wiley \&
  Sons.

\bibitem[{Gilraine \textit{et~al.}(2020)Gilraine, Gu and
  McMillan}]{gilraine2020new}
\textsc{Gilraine, M.}, \textsc{Gu, J.} and \textsc{McMillan, R.} (2020).
  \textit{A new method for estimating teacher value-added}. Tech. rep.,
  National Bureau of Economic Research.

\bibitem[{Gu and Koenker(2017)}]{gu2017unobserved}
\textsc{Gu, J.} and \textsc{Koenker, R.} (2017). Unobserved heterogeneity in
  income dynamics: An empirical bayes perspective. \textit{Journal of Business
  \& Economic Statistics}, \textbf{35}~(1), 1--16.

\bibitem[{Jackson(2018)}]{jackson2018test}
\textsc{Jackson, C.~K.} (2018). What do test scores miss? the importance of
  teacher effects on non--test score outcomes. \textit{Journal of Political
  Economy}, \textbf{126}~(5), 2072--2107.

\bibitem[{Jacob and Lefgren(2005{\natexlab{a}})}]{jacob2005principle}
\textsc{Jacob, B.} and \textsc{Lefgren, L.} (2005{\natexlab{a}}).
  \textit{Principals as {{Agents}}: {{Subjective Performance Measurement}} in
  {{Education}}}. Tech. Rep. w11463, {National Bureau of Economic Research},
  {Cambridge, MA}.

\bibitem[{Jacob and Lefgren(2005{\natexlab{b}})}]{jacob2005parents}
\textsc{---} and \textsc{---} (2005{\natexlab{b}}). What do parents value in
  education? an empirical investigation of parents' revealed preferences for
  teachers.

\bibitem[{Kane \textit{et~al.}(2008)Kane, Rockoff and Staiger}]{kane2008does}
\textsc{Kane, T.~J.}, \textsc{Rockoff, J.~E.} and \textsc{Staiger, D.~O.}
  (2008). What does certification tell us about teacher effectiveness?
  {E}vidence from {New York City}. \textit{Economics of Education Review},
  \textbf{27}~(6), 615--631.

\bibitem[{Kane and Staiger(2008)}]{kane2008estimating}
\textsc{---} and \textsc{Staiger, D.~O.} (2008). \textit{Estimating teacher
  impacts on student achievement: An experimental evaluation}. Tech. rep.,
  National Bureau of Economic Research.

\bibitem[{Kwon(2023)}]{kwon2023optimalshrinkageestimationfixed}
\textsc{Kwon, S.} (2023). Optimal shrinkage estimation of fixed effects in
  linear panel data models.

\bibitem[{Lehmann and Romano(2008)}]{lehmann1986testing}
\textsc{Lehmann, E.~L.} and \textsc{Romano, J.~P.} (2008). \textit{Testing
  statistical hypotheses}, vol.~3. Springer.

\bibitem[{Mulhern(2023)}]{mulhern2023beyond}
\textsc{Mulhern, C.} (2023). Beyond teachers: Estimating individual school
  counselors’ effects on educational attainment. \textit{American economic
  review}, \textbf{113}~(11), 2846--2893.

\bibitem[{Robbins(1956)}]{robbins56}
\textsc{Robbins, H.} (1956). An empirical {B}ayes approach to statistics. In
  \textit{Proceedings of the Third Berkeley Symposium on Mathematical
  Statistics and Probability}, University of California Press: Berkeley,
  vol.~I.

\bibitem[{Rose \textit{et~al.}(2022)Rose, Schellenberg and
  Shem-Tov}]{rose2022effects}
\textsc{Rose, E.~K.}, \textsc{Schellenberg, J.~T.} and \textsc{Shem-Tov, Y.}
  (2022). \textit{The effects of teacher quality on adult criminal justice
  contact}. Tech. rep., National Bureau of Economic Research.

\bibitem[{Stigler(1990)}]{stigler19901988}
\textsc{Stigler, S.~M.} (1990). The 1988 neyman memorial lecture: a galtonian
  perspective on shrinkage estimators. \textit{Statistical Science}, pp.
  147--155.

\bibitem[{Sullivan(2001)}]{sullivan2001note}
\textsc{Sullivan, D.~G.} (2001). A note on the estimation of linear regression
  models with heteroskedastic measurement errors. \textit{Available at SSRN
  295567}.

\bibitem[{Tsiatis(2006)}]{tsiatis2006semiparametric}
\textsc{Tsiatis, A.~A.} (2006). \textit{Semiparametric theory and missing
  data}, vol.~4. Springer.

\bibitem[{van~der Vaart(2000)}]{van2000asymptotic}
\textsc{van~der Vaart, A.~W.} (2000). \textit{Asymptotic statistics}, vol.~3.
  Cambridge university press.

\bibitem[{Walters(2024)}]{walters2024empirical}
\textsc{Walters, C.} (2024). Empirical bayes methods in labor economics. In
  \textit{Handbook of Labor Economics}, vol.~5, Elsevier, pp. 183--260.

\bibitem[{Warnick \textit{et~al.}(2025)Warnick, Light and
  Yim}]{warnick2024instructor}
\textsc{Warnick, M.}, \textsc{Light, J.} and \textsc{Yim, A.} (2025).
  Instructor value-added in higher education.

\bibitem[{Wu and Yang(2020)}]{wu2020polynomial}
\textsc{Wu, Y.} and \textsc{Yang, P.} (2020). Polynomial methods in statistical
  inference: Theory and practice. \textit{Foundations and
  Trends{\textregistered} in Communications and Information Theory},
  \textbf{17}~(4), 402--586.

\bibitem[{Xie(2025)}]{xie2025automatic}
\textsc{Xie, T.} (2025). \textit{Automatic Inference for Value-Added
  Regressions}. Tech. rep.

\end{thebibliography}

\newpage 

\appendix 

\counterwithin{equation}{section}

\section{Proof of \cref{lemma:implication}}
We restate and prove the following result. 

\lemmaimplication*

\begin{proof}
	\begin{enumerate}
		\item Precision independence implies precision independence in the first two
		moments. Thus \cref{as:strong_low_level} implies \cref{as:strong_high_level}(1)
		by part (2) of this lemma.
		For exogeneity, write \[
			\mu^*(X_i, \sigma_i) = (1-w(\sigma_i)) \mu + w(\sigma_i) X_i
		\]
		for $\mu = \E[\mu_i]$
		By assumption, $\sigma_i \indep (\mu_i, \eta_i)$.		Then \begin{align*}
		\E[\eta_i\mu^*(X_i, \sigma_i)] &= \E[(1-w(\sigma_i))]  \mu \underbrace{\E
		[\eta_i]}_0 + \E
			[w(\sigma_i) (\mu_i + \sigma_i \epsilon_i) \eta_i] \quad \epsilon_i \sim \Norm
			(0,1)  \\
			&= \E[w(\sigma_i)] \underbrace{\E[\mu_i \eta_i]}_0 + \E[w(\sigma_i) \sigma_i] 
			\underbrace{\E
			[\epsilon_i]}_{0}
			\E[\eta_i] = 0.
		\end{align*}

		\item We compute \begin{align*}
		\cov(\mu_i, \mu^*(X_i, \sigma_i)) &= \E\cov(\mu_i, \mu^*(X_i, \sigma_i) \mid
		\sigma_i) + \cov\pr{
			\underbrace{\E[\mu_i \mid \sigma_i]}_{\text{constant } \mu}, \E[\mu^*
			(X_i,\sigma_i) \mid
			\sigma_i]
		} \\
		&= \E\bk{w(\sigma_i) \cov(\mu_i, X_i \mid \sigma_i)} \\
		&= \E\bk{w(\sigma_i) \var(\mu_i \mid \sigma_i)} \\
		&= \E\bk{\frac{1}{\sigma_\mu^2 + \sigma_i^2}} \sigma_\mu^4. 
		\end{align*}
		and \begin{align*}
		\var(\mu^*(X_i, \sigma_i)) &= \E \var(\mu^*(X_i, \sigma_i) \mid \sigma_i) +
		\var\pr{
			\underbrace{\E\bk{\mu^*(X_i, \sigma_i) \mid \sigma_i}}_{\text{constant }\mu}
		} \\
		&= \E\bk{
			w^2(\sigma_i) \var(X_i \mid \sigma_i)
		} \\
		&= \E\bk{
		\pr{\frac{\sigma_\mu^2}{\sigma_\mu^2 + \sigma_i^2}}^2	(\sigma_\mu^2 +
		\sigma_i^2)
		}\\
		&= \E\bk{\frac{1}{\sigma_\mu^2 + \sigma_i^2}} \sigma_\mu^4.
		\end{align*}
		Thus \[
			\cov(\mu_i, \mu^*(X_i, \sigma_i))/
\var(\mu^*(X_i, \sigma_i))  = 1.
		\]

\item Finally, we have that $\E[\eta_i \mid \mu_i, \sigma_i] = 0$ by assumption. Thus, \begin{align*}
\E[\eta_i \mu^*(X_i, \sigma_i)] &= \E\bk{
	\E[\eta_i \mu^*(X_i, \sigma_i) \mid \mu_i, \sigma_i]
}\\ &= \E\bk{
	\E\bk{
		\eta_i \br{
			(1-w(\sigma_i)) \mu + w(\sigma_i) (\mu_i + \sigma_i \epsilon_i)
		} \mid \mu_i, \sigma_i
	}
} \\
&= \E\bk{\sigma_i \E[\eta_i \epsilon_i \mid \mu_i, \sigma_i]} \tag{$\E[\eta_i \mid
\sigma_i, \mu_i] = 0$}
\\&= 0.
\end{align*}
where the last step follows since $\epsilon_i \mid (Y_i, \mu_i, \sigma_i)$ is mean zero.
	\end{enumerate}
\end{proof}

\section{Proof of \cref{thm:eff}}
\label{asec:eff}

We first carefully define $P_0, \mathcal P_0, P, \mathcal P$. Given $P_0$ a joint
distribution on the complete variables $(Y_i, \mu_i, \sigma_i)$,
we define $P^\obs = P^{\obs}(P_0)$ as the induced distribution on the observed variables $
(Y_i, X_i, \sigma_i)$. Let $\mathcal P_0$ be the set of distributions $P_0$ that satisfy
\cref{basicassp} and are dominated by some $\sigma$-finite product measure $\lambda_0 =
\lambda_Y
\otimes \lambda_{\R} \otimes \lambda_\mu \otimes \lambda_\sigma$, where $\lambda_\R$ is
the Lebesgue measure and (iii) the support of $\mu \mid Y, \sigma$
under $P_0$ contains a nonempty interval almost surely. Similarly, let $\mathcal P =
\br{P^ {\obs} (P_0)
:P_0 \in \mathcal
P_0}$, where $\mathcal P$ is dominated by $\lambda = \lambda_Y \otimes \lambda_\R \otimes
\lambda_\sigma$. Let $\mathcal Z \subset \R^3$ denote the set of values that $ (Y_i, X_i,
\sigma_i)$ takes. For a given $P \in \mathcal P$, define $L^2(P)$ as the Hilbert space of
square-integrable functions $\mathcal Z \to \R$ under $P$ and $L^2_0(P)$ as the Hilbert
space of square-integrable and mean-zero functions $\mathcal Z \to \R$ under $P$.

We recall the following definitions from semiparametric theory \citep[see, e.g., chapter
25 of] [] {van2000asymptotic}.

\begin{defn}[Parametric submodel]
For a given $P \in \mathcal P$, a \emph{smooth parametric submodel} is a set $\br{P_t : t
\in [0, \epsilon)} \subset \mathcal P$ that is differentiable in quadratic mean at $t=0$
and
$P_{t=0} = P$: For some measurable function $g: \mathcal Z \to \R$,\[
\lim_{t \downarrow 0} \int \br{\frac{1}{t} (\sqrt{p_t} - \sqrt{p}) - \frac{1}{2} g 
\sqrt{p}}^{2} d\lambda = 0,
\] 
where $p_t = dP_t/d \lambda$ and $p = dP/d\lambda$ are densities with respect to the
dominating measure. We refer to $g$ as the \emph{score} of the submodel.
\end{defn}

\begin{defn}[Tangent space]
For a given $P \in \mathcal P$, the \emph{tangent set} of $\mathcal P$ at $P$ is defined
as \[
\mathcal T(P) \equiv \br{g \in L_0^2(P) : \text{There exists a smooth parametric
submodel with score $g$}} \subset L^2_0(P).
\]
The \emph{tangent space} $\Tspace(P)$ at $P$ is defined as the closure of the
linear span of $\mathcal
T(P)$ with respect to $L^2_0(P)$.
\end{defn}

\begin{defn}[Regularity and asymptotic linearity]
For a given $P \in \mathcal P$ and a given parameter $\theta(P) \in \R$, an estimator
$\hat\theta_n$ is \emph{regular} if there exists a distribution $L$ such that along all
smooth parametric submodels $P_t$, \[
\sqrt{n} (\hat\theta_n - \theta(P_{1/\sqrt{n}})) \underset{P_{1/\sqrt{n}}}{\dto} L.
\]
$\hat\theta_n$ is $\emph{asymptotically linear}$ at $P$ if there exists some
\emph{influence function} $\psi \in L_0^2(P)$ such that \[
\sqrt{n} (\hat\theta_n - \theta(P)) =  \frac{1}{\sqrt{n}} \sum_{i=1}^n
\psi(Y_i, X_i, \sigma_i; \theta(P)) + o_P(1).
\]
\end{defn}

Finally, we recall Definition 2.2 from \citet{chen2018overidentification}.

\begin{defn}[Local just identification]
For a given $P \in \mathcal P$, if $\Tspace(P) = L^2_0(P)$, we say $P$ is locally just
identified by $\mathcal P$.
\end{defn}

Loosely speaking, local just identification at $P$ means that $P$ can be perturbed in any
direction within $\mathcal P$ and the model $\mathcal P$ is consistent with any (local)
parametric model at $P$. Thus, there is no information that an analyst could exploit by
imposing the model $\mathcal P$, since $\mathcal P$ makes no restrictions locally at $P$.
This is the semiparametric analogue to just identification in parametric GMM models, where
there are no additional moment conditions to exploit, and all GMM weightings yield the
same estimator. Indeed, \citet{chen2018overidentification} (Theorem 3.1(i)) show that
when $P$ is locally just identified, all RAL estimators are asymptotically equivalent.

This section verifies the following theorem, under a technical assumption stated
immediately after.

\thmeff*

\begin{proof}
The proof follows by applying Example 25.35 in \citet{van2000asymptotic}. 

Fix $P_0
	\in \mathcal P_0$, let $Q_0$ be the corresponding distribution of $(Y_i, \mu_i,
	\sigma_i)$ under $P_0$ and let $\mathcal Q_0$ collect all such distributions as $P_0
$ ranges over $\mathcal P_0$. First observe that the tangent space at $Q_0$ relative to
	$\mathcal Q_0$ \[
		\bar{\mathcal T}(Q_0) = L^2_0(\mathcal Q_0).
	\]

	Let $\bm X = (Y_i, X_i, \sigma_i) \sim P$ and $\bm Z = (Y_i, \mu_i, \sigma_i)
	\sim Q_0$. Note
	that the density of $\bm X \mid \bm Z$ (with respect to the dominating measure
	$\delta_{y} \otimes \lambda_\R \otimes \delta_{\sigma}$) is \[p (\bm X \mid \bm Z =
	(y, \mu, \sigma)) =
	\frac{1}
	{\sqrt{2\pi} \sigma} \exp\pr{ -\frac{1}{2\sigma^2} (x - \mu)^2 },\] which is
		exponential family. The score space for $P$ consists of the functions \[ (A_{Q_0}
		b)(\bm x) \equiv \E_{Q_0}[b(\bm Z) \mid \bm X = \bm x = (y, x, \sigma)] = 
		\frac{\int_ {-\infty}^
		{\infty} b (y, \mu,
		\sigma) \frac{1}
	{\sqrt{2\pi} \sigma} \exp\pr{
		-\frac{1}{2\sigma^2} (x - \mu)^2 
	}\, dQ_0(\mu \mid y, \sigma)}{
		\int_{-\infty}^
		{\infty}  \frac{1}
	{\sqrt{2\pi} \sigma} \exp\pr{
		-\frac{1}{2\sigma^2} (x - \mu)^2 
	}\, dQ_0(\mu \mid y, \sigma)
	}
	\]
	for $b$ in the tangent space $\bar {\mathcal T}(Q_0)$. Following Example 25.35 in 
	\citet{van2000asymptotic}, we show that such scores are dense in $L^2_0(P)$. 

	It suffices to show that the closure of the range of the \emph{score operator}
	$A_{Q_0}$ is $L^2_0 (P)$, since $\bar {\mathcal T}(Q_0) = L^2_0(Q_0)$. This is further
	equivalent to showing the orthocomplement of the kernel $N(A_{Q_0}^*)$ is equal to
	$L^2_0(P)$. Thus it suffices to show that the kernel $N(A_{Q_0}^*)$ is trivial: That
	is, \[
		0 = (A_\eta^* g)(\bm z) = \E[g(\bm X) \mid \bm Z = \bm z = (y, \mu, \sigma)]
		\text{ $Q_0(\cdot \mid y, \sigma)$-a.s.} \implies g (y, \cdot, \sigma) = 0 \text{
		a.e.}
	\]
	By assumption, the support of $\mu \mid Y, \sigma$ under $Q_0$ contains an interval
	a.s., and thus the above display is true, for $Q_0$-almost all $(Y, \sigma)$, by the
	completeness of Gaussian location models (Theorem 4.3.1 in
	\citet{lehmann1986testing}). This shows that members of $N(A_ {Q_0}^*)$ are almost
	surely zero, and this completes the proof for the first statement.

	Theorem 3.1(i) in \citet{chen2018overidentification} immediately implies (1), since (a)
the tangent set $\mathcal T(P)$ is closed under linear combinations and (b) $\hat\beta$ is
an RAL estimator. For (2), since $\hat\beta$ is an RAL estimator, the projection of
$\hat\beta$'s influence function onto $\Tspace(P)$ is the efficient influence function.
However, since $\Tspace(P) = L^2_0(P)$, the projection of $\hat\beta$'s influence function
onto $\Tspace(P)$ is itself. This observation implies (2).
\end{proof}

\section{Proof of \cref{thm:minimax}}

This section restates and proves the following theorem. 
\thmminimax*

\begin{rmksq}
\Cref{lemma:lecam} provides a construction linking
functional estimation in Gaussian white noise models to estimation of $T(Q)$ and derives a
lower bound using Le Cam's two-point method, given two priors $G_{-1}, G_1$ of $\mu$.
\Cref{thm:duality} shows that certain worst-case choices of $G_{-1}, G_1$ connect to the
problem of uniform approximation by polynomials. \Cref{thm:analytic} is a result in
approximation theory that shows that functions are well-approximated by polynomials if and
only if they are analytic.

\Cref{thm:duality,lemma:lecam} links the minimax lower bound to polynomial approximation.
\Cref{thm:analytic} shows that if $f$ is not analytic, then the approximation rate cannot
decay exponentially. The proof here puts things together and verifies that, \emph{therefore}, the
minimax
rate cannot be polynomial. 
\end{rmksq}

\begin{proof}

By \cref{lemma:lecam} and \cref{thm:duality}, we know that for any integer $k \ge 1$, we have that \[
R_n(\mathcal Q_0, f) \gtrsim_{\norm{f}_\infty} E_k(f)^2 \pr{1 - C
\exp\pr{
	-\frac{k+1}{2}(\log(k+1)-1) + \frac{1}{2}\log n
}
}
\]
where $E_k(f)$ is defined in \cref{thm:duality}. If $f$ is not analytic, by 
\cref{thm:analytic}, we have that\footnote{Note that $E_k(f)$ is bounded and decreasing,
and so $E_k(f)^{1/k} \le E_1(f)^{1/k} \to 1.$} \[\limsup_{k\to\infty} E_k(f)^{1/k} = 1.
\numberthis \label{eq:analytic_rate}
\]
Hence there exists a subsequence $(k_\ell)$ where $E_{k_\ell}(f)^{1/k_\ell} \to 1$ as $\ell
\to \infty$. Note that we can take another subsequence $(k_n)$ such that (i) $k_n
\rateeq \log n$ and
(ii) infinitely many elements from $(k_\ell)$ features in $(k_n)$. 

For such a $k_n$, \[
R_n(\mathcal Q_0, f) \gtrsim_{\norm{f}_\infty} E_{k_n}(f)^2.
\]

Suppose, for contradiction, there is some $\alpha$ for which \[
\limsup_{n\to\infty} n^{\alpha} R_n(\mathcal Q_0, f) < \infty.
\]
Then \begin{align*}
\limsup_{n\to\infty} n^{\alpha} E_{k_n}(f)^2  \le \infty.
\end{align*}
This means that for some subsequence of $(k_n)$ that contains infinitely many  elements of
$k_\ell$, $(k_ {n_m})$, we have that $
\sup_m |n_m^\alpha E_{k_{n_m}(f)}^{2}| < M^2 < \infty. 
$
However, \begin{align*}
E_{k_{n_m}}(f)^{1/k_{n_m}} &= (n_m^{\alpha/2}E_{k_{n_m}}(f) )^{1/k_{n_m}} n_m^{-
\frac{\alpha}{2k_{n_m}}} \\
&\le M^{1/k_{n_m}} \exp\pr{-\frac{\alpha}{2k_{n_m}} \log n_m} \\
&\overset{m \to \infty}{\longrightarrow} \exp(-\alpha c) < 1
\end{align*}
where $
c \equiv 2\lim_{m\to\infty} \frac{\log n_m}{k_{n_m}} > 0
$
since $k_{n_m} \rateeq \log n_m$. This contradicts \eqref{eq:analytic_rate}.
\end{proof}

\begin{lemma}
\label{lemma:lecam}
In this problem, for any $G_{-1}, G_1$ supported on $[-1,1]$, \[
R_n(\mathcal Q_0, f) \gtrsim_{\norm{f}_\infty} \pr{\E_{G_{1}}[f
(\mu)]
- \E_{G_{-1}}[f(\mu)]}^2 \pr{1-\frac{1}{2\sqrt{2}} \sqrt{n \chi^2(G_{-1}\star \Norm
(0,1), G_1\star \Norm(0,1))}}
\]
where $\chi^2(Q_1, Q_2)$ is the $\chi^2$-divergence between $Q_1$ and $Q_2$. 
\end{lemma}

\begin{proof}
Fix some distributions $G_{-1}, G_1$ supported on $[-1,1]$. Let $G_{-1}^*, G_1^*$ be the
mixture \[
G_j^* = \frac{1}{2} G_{j} + \frac{1}{2} \Unif[-1,1].
\]
Note that $\var_{G_j^*}[f(\mu)] \ge \frac{V(f) }{2}$. Choose distributions such that
$\sigma = 1$ almost surely. 

Consider induced $Q_0, Q_1$
where
under $Q_j$, \begin{align*}
Y &\sim_{Q_j} \mathrm{Rademacher}(1/2) \text{ for $j = 0,1$}\\ 
\mu \mid Y &\sim_{Q_0} G_{-1}^* \\ 
\mu \mid Y = -1 &\sim_{Q_1} G_{-1}^* \\ 
\mu \mid Y = 1 &\sim_{Q_1} G_1^*.
\end{align*}

Thus, $Q_0, Q_1 \in \mathcal Q$. Note that \[
T(Q_0) = \frac{
	\frac{1}{2}\E_{G_{-1}^*}[f(\mu)] - \frac{1}{2}\E_{G_{-1}^*}[f(\mu)]
}{\var_{Q_0}(f(\mu))} = 0
\]
and \begin{align*}
T(Q_1) &= \frac{
	\frac{1}{2}\E_{G_{1}^*}[f(\mu)] - \frac{1}{2}\E_{G_{-1}^*}[f(\mu)]
}{\var_{G_1^*}(f(\mu))} \\&= \frac{1}{2\var_{G_1^*}(f(\mu)) }\pr{\E_{G_{1}^*}[f
(\mu)]
- \E_{G_{-1}^*}[f(\mu)]} \\
&=\frac{1}{4\var_{G_1^*}(f(\mu)) }\pr{\E_{G_{1}}[f(\mu)]
- \E_{G_{-1}}[f(\mu)]}
\end{align*}

Let $P(Q) = P^\obs(Q)$ denote distribution of the observed data induced by $Q \in \mathcal
Q_0$. Recall that Le Cam's two-point method states that for some absolute $c > 0$ \[
R_n(\mathcal Q_0, f) \ge c (T(Q_1) - T(Q_0))^2 \pr{1- \mathrm{TV}\pr{ P(Q_1)^{\otimes n},
P(Q_0)^{\otimes n} }}.
\]

Note that \begin{align*}
&1- \mathrm{TV}\pr{
	P(Q_1)^{\otimes n}, P(Q_0)^{\otimes n}
} \\ &\ge 1-\frac{1}{\sqrt{2}}\sqrt{\mathrm{KL}\pr{P(Q_0)^{\otimes
n}, P(Q_1)^{\otimes n}}} \\
&= 1-\frac{1}{\sqrt{2}}\sqrt{n\mathrm{KL}\pr{P(Q_0) , P(Q_1) }} \tag{Tensorization of $
\mathrm{KL}$} \\
&= 1-\frac{1}{\sqrt{2}}\sqrt{\frac{n}{2}\pr{\mathrm{KL}\pr{G_{-1}^* \star \Norm(0,1), G_
{-1}^* \star\Norm(0,1)} + \mathrm{KL}\pr{G_{-1}^* \star \Norm(0,1), G_
{1}^*\star \Norm(0,1)}}} \tag{Conditional KL given $Y=0,Y=1$} \\ 
&= 1- \frac{1}{2}\sqrt{n \mathrm{KL}(G_{-1}^* \star \Norm(0,1), G_1^* \star \Norm(0,1))} \\ 
&= 1-\frac{1}{2} \sqrt{n \chi^2(G_{-1}^*\star \Norm(0,1), G_1^*\star \Norm(0,1))} \\ 
&\ge 1-\frac{1}{2\sqrt{2}} \sqrt{n \chi^2(G_{-1}\star \Norm(0,1), G_1\star \Norm(0,1))} 
\tag{By convexity: $\chi^2(G_{-1}^* \star\Norm(0,1), G_1^*\star \Norm(0,1)) \le \frac{1}{2} \chi^2
	\pr{G_ {-1} \star \Norm(0,1), G_1 \star \Norm(0,1)}$
}
\end{align*}

Putting things together, we have that 
\begin{align*}
&R_n(\mathcal Q_0, f) \\
&\gtrsim  \frac{1}{\var_{G_1}(f(\mu))^2} \pr{\E_{G_{1}}[f
(\mu)]
- \E_{G_{-1}}[f(\mu)]}^2 \pr{1-\frac{1}{2\sqrt{2}} \sqrt{n \chi^2(G_{-1}\star \Norm
(0,1), G_1\star \Norm(0,1))}} \\ 
&\gtrsim_{\norm{f}_\infty} \pr{\E_{G_{1}}[f
(\mu)]
- \E_{G_{-1}}[f(\mu)]}^2 \pr{1-\frac{1}{2\sqrt{2}} \sqrt{n \chi^2(G_{-1}\star \Norm
(0,1), G_1\star \Norm(0,1))}}.
\end{align*}
\end{proof}

\begin{theorem}
\label{thm:duality}
Fix bounded and measurable $f : [-1,1]\to \R$, let \[
E_k(f) = \inf_{a_{0:k} \in \R^{k+1}} \norm{f(x) - a_0 - a_1 x^1 - \cdots - a_k x^k}_\infty
\]
be the best approximation error of $f$ via polynomials. Then, for universal $c, C >
0$, for any $k \ge 1$, there exist $G_ {-1}, G_1$
such that \[
|\E_{G_1}[f(\mu)] - \E_{G_{-1}}[f(\mu)]| \ge c E_k(f)
\]
and \[
\chi^2\pr{G_{1} \star \Norm(0,1), G_{-1} \star \Norm(0,1)} \le C \exp\pr{-(k+1)(\log
(k+1)-1)}.
\]
\end{theorem}

\begin{proof}
By Theorem 3.3.3 in \citet{wu2020polynomial}, there exists $G_{1}, G_{-1}$ that matches
the first $k$ moments and \[
\chi^2\pr{G_{1} \star \Norm(0,1), G_{-1} \star \Norm(0,1)} \le C \exp\pr{-(k+1)(\log
(k+1)-1)}.
\]
We can maximize $\E_{G_1}[f(\mu)] - \E_{G_{-1}}[f (\mu)]$ over moment-matching
distributions $G_{1}, G_{-1}$. The dual of this linear program is the problem of uniform
polynomial approximation. By (2.9) in \citet{wu2020polynomial}, the optimal value of this
problem is of the form $c E_k(f)$, where $E_k(f)$ is the approximation error.
\end{proof}

\begin{theorem}
\label{thm:analytic}
A function $f: [-1,1]\to \R$ is analytic if and only if \[
\limsup_{k \to \infty} E_k(f)^{1/k} < 1. 
\]
\end{theorem}
\begin{proof}
This is Theorem 8.1 in \cite{devore1993constructive}.
\end{proof}

\section{Details for the empirical application to \citet{bau2020teacher}}
\label{sec:details_for_the_empirical_application}

Replication code for \citet{bau2020teacher} outputs a set of noisy classroom effects $X_c$
and estimates the variance of class effects $\hat\sigma_c^2$, school effects
$\hat\sigma_s^2$, teacher effects $\hat\sigma_\mu^2$, and
idiosyncratic variation in student test scores $\hat\sigma_e^2$. 

Within a teacher $i$, the classroom effects are uncorrelated, have common mean $\mu_i$,
and variances $v_c = \sigma_c^2 + \frac{\sigma^2_e}{n_c}$, for $n_c$ the size of a
class.
Following footnote 23 of \citet{bau2020teacher}, a teacher effect is the
precision-weighted average of classroom effects for that teacher, and follows \[
	X_i = \sum_{c} w_c X_c \sim \Norm(\mu_i, \sigma_i^2) \quad w_c \propto 1/v_c
\]
where the sum is over classrooms that are taught by teacher $i$. The variance of the
teacher effect \[
	\var(X_i \mid \mu_i) = \pr{\sum_c \frac{1}{v_c}}^{-1}
\]
is the harmonic sum of the variance of the classroom effects. We thus treat $\sigma_i^2$
as $(\sum_c 1/\hat v_c)^{-1}$ where $\hat v_c$ replaces $\sigma_c^2, \sigma_e^2$ with
their estimated counterparts produced by \citet{bau2020teacher}. 

These effects are computed for each test score subject (math, English, Urdu). Following
\citet{bau2020teacher}, we treat \[
	X_i = \frac{1}{3} \pr{X_i^{\text{math}} + X_i^{\text{English}} + X_i^{\text{Urdu}}},
\]
as the measure of teacher value-added, and its variance $\sigma_i^2$ as $1/3$ times the
arithmetic mean of the variances for each subject.

\end{document}